\documentclass[aos,preprint]{imsart2} 
\RequirePackage[OT1]{fontenc}
\usepackage[colorlinks,citecolor=blue,urlcolor=blue]{hyperref}
\usepackage{url}
\usepackage{amsthm,amsfonts,amstext,amsmath,amssymb}
\usepackage{bbm}
\usepackage{mathrsfs}
\usepackage{graphicx}
\usepackage{caption}
\usepackage{subcaption}
\usepackage{verbatim}
\usepackage{enumerate}
\usepackage{color}
\RequirePackage[numbers]{natbib}

\arxiv{1509.07566}

\startlocaldefs
\numberwithin{equation}{section}
\theoremstyle{plain}
\newtheorem{theorem}{Theorem}[section]
\newtheorem{lemma}[theorem]{Lemma}

\newtheorem{corollary}[theorem]{Corollary}
\endlocaldefs

\newcommand\numberthis{\addtocounter{equation}{1}\tag{\theequation}}
\DeclareMathOperator*{\argmin}{arg\,min}
\newcommand{\pfa}{\Pro_{\rm FA}}
\newcommand{\pmd}{\Pro_{\rm MD}}

\newcommand{\LR}{\mathsf{L}}                   
\newcommand{\Exp}{\mathsf{E}}                  
\newcommand{\Pro}{\mathsf{P}}                  
\newcommand{\Hyp}{\mathsf{H}}                  
\newcommand{\f}{\mathsf{f}}
\newcommand{\LLR}{{\rm LLR}}
\newcommand{\ind}[1]{\mathbbm{1}_{\{#1\}}}     

\newcommand{\george}[1]{{}}

\begin{document}

\begin{frontmatter}
\title{Detecting Sparse Mixtures: Rate of Decay of Error Probability \thanksref{T1}}
\runtitle{Detecting Sparse Mixtures}

\begin{aug}
\author{\fnms{Jonathan G.} \snm{Ligo}\thanksref{m1}\ead[label=e1]{ligo2@illinois.edu}},
\author{\fnms{George V.} \snm{Moustakides}\thanksref{m2}\ead[label=e2]{moustaki@upatras.gr}}
\and
\author{\fnms{Venugopal V.} \snm{Veeravalli}\thanksref{m1}
\ead[label=e3]{vvv@illinois.edu}
}

\thankstext{T1}{Supported by the US National Science Foundation under grants CIF\,1514245 and CIF\,1513373.  This paper was partially presented at the 41st IEEE International Conference on Acoustics, Speech and Signal Processing (ICASSP) 2016 as \cite{icassp}.}
\runauthor{J. G. Ligo et al.}

\affiliation{University of Illinois at Urbana-Champaign\thanksmark{m1}, University of Patras\thanksmark{m2} and Rutgers University\thanksmark{m2}}

\address{Coordinated Science Laboratory\\ and\\
Department of Electrical and\\ 
Computer Engineering\\
University of Illinois at\\
Urbana-Champaign\\
Urbana, IL 61801, USA\\
\printead{e1}\\
\phantom{E-mail:\ }\printead*{e3}}

\address{Department of Electrical and\\
Computer Engineering\\
University of Patras \\
26500 Rio, Greece\\
and\\
Department of Computer Science\\
Rutgers University \\
New Brunswick, NJ 08854, USA\\
\printead{e2}
}
\end{aug}

\begin{abstract}
We study the rate of decay of the probability of error for distinguishing between a sparse signal with noise, modeled as a sparse mixture, from pure noise. This problem has many applications in signal processing, evolutionary biology, bioinformatics, astrophysics and feature selection for machine learning. 
We let the mixture probability tend to zero as the number of observations tends to infinity and derive oracle rates at which the error probability can be driven to zero for a general class of signal and noise distributions via the likelihood ratio test. In contrast to the problem of detection of non-sparse signals, we see the log-probability of error decays sublinearly rather than linearly and is characterized through the $\chi^2$-divergence rather than the Kullback-Leibler divergence for ``weak'' signals and can be independent of divergence for ``strong'' signals. Our contribution is the first characterization of the rate of decay of the error probability for this problem for both the false alarm and miss probabilities.
\end{abstract}



\end{frontmatter}

\section{Introduction}
We consider the problem of detecting a sparse signal in noise, modeled as a mixture, where the unknown sparsity level decreases as the number of samples collected increases. Of particular interest is the case where the unknown signal strength relative to the noise power is very small. This problem has many natural applications. In signal processing, applications include detecting a signal in a multi-channel system \cite{dobrushin,ingster} and detecting covert communications \cite{donoho}. In evolutionary biology, the problem manifests in the reconstruction of phylogenetic trees in the multi-species coalescent model \cite{mosselroch}. In bioinformatics, the problem arises in the context of determining gene expression from gene ontology datasets \cite{goeman}. In astrophysics, detection of sparse mixtures is used to compare models of the cosmic microwave background to observed data \cite{cayon}. Also, statistics developed from the study of this problem have been applied in machine learning to anomaly detection on graphs \cite{saligrama} and high-dimensional feature selection when useful features are rare and weak \cite{feature}. 

Prior work on detecting a sparse signal in noise has been primarily focused on Gaussian signal and noise models, with the goal of determining the trade-off in signal strength with sparsity required for detection with vanishing probability of error. In contrast, this work considers a fairly general class of signal and noise models. Moreover, in this general class of sparse signal and noise models, we provide the first analysis of the rate at which the false alarm {(Type-I)} and miss detection {(Type-II)} error probabilities vanish with sample size. We also provide simple to verify conditions for detectability, { which are derived using simpler tools than previously used}. In the problem of testing between $n$ i.i.d. samples from two known distributions, it is well known that the rate at which the error probability decays is $e^{-c n}$ for some constant $c>0$ bounded by the Kullback-Leibler divergence between the two distributions \cite{cover,dz}. In this work, we show for the problem of detecting a sparse signal in noise that the error probability for an oracle detector decays at a slower rate determined by the sparsity level and the $\chi^2$-divergence between the signal and noise distributions, with different behaviors possible depending on the signal strength. In addition to determining the optimal trade-off between signal strength and sparsity for consistent detection, an important contribution in prior work has been the construction of adaptive (and, to some extent, distribution-free) tests that achieve the optimal trade-off without knowing the model parameters \cite{castro,caijengjin,caiwu,donoho,ingster,jager,walther}. We discuss prior work in more detail in Sec.~\ref{sec:rw}. However, the adaptive tests that have been proposed in these papers are not amenable to an analysis of the {\em rate} at which the error probability goes to zero. We show that in a Gaussian signal and noise model that an adaptive test based on the sample maximum has miss detection probability that vanishes at the optimal rate when the sparse signal is sufficiently strong.
\section{Problem Setup}
Let $\{\f_{0,n}(x)\}, \{\f_{1,n}(x)\}$ be sequences of probability density functions (PDFs) for real valued random-variables.

We consider the following sequence of composite hypothesis testing problems with sample size $n$, called the (sparse) \emph{mixture detection problem}:
\begin{align}
\Hyp_{0,n}:&~~ X_1, \ldots, X_n \sim \f_{0,n}(x) \text{ i.i.d. (null)}\\
\Hyp_{1,n}:&~~ X_1, \ldots, X_n \sim (1-\epsilon_n) \f_{0,n}(x) + \epsilon_n \f_{1,n}(x) \text{ i.i.d. (alternative)} \label{eq:altdef}
\end{align}
where $\{\f_{0,n}(x)\}$ is known, $\{\f_{1,n}(x)\}$ is from some known family $\mathcal{F}$ of sequences of PDFs, and $\{\epsilon_n\}$ is a sequence of positive numbers such that $\epsilon_n \to 0$. We will also assume $n \epsilon_n \to \infty$ so that a typical realization of the alternative is distinguishable from the null. 

Let $\Pro_{0,n},\Pro_{1,n}$ denote the probability measure under $\Hyp_{0,n},\Hyp_{1,n}$ respectively, and let $\Exp_{0,n},\Exp_{1,n}$ be the corresponding expectations, with respect to the particular $\{\f_{0,n}(x)\}$, $\{\f_{1,n}(x)\}$ and $\{\epsilon_n\}$. When convenient, we will drop the subscript $n$. Let $\LR_n \triangleq \frac{\f_{1,n}(x)}{\f_{0,n}(x)}$. When $\f_{0,n}(x) = \f_0(x)$ and $\f_{1,n}(x) = \f_0(x-\mu_n)$, we say that the model is a \emph{location model}. {For the purposes of presentation, we will assume that $\{\mu_n\}$ is a positive and monotone sequence.} When $\f_{0}(x)$ is a standard normal PDF, we call the location model a \emph{Gaussian location model}. The distributions of the alternative in a location model are described by the set of sequences $\{ (\epsilon_n,\mu_n)\}$. 

The location model can be considered as one where the null corresponds to pure noise, while the alternative corresponds to a sparse signal (controlled by $\epsilon_n$), with signal strength $\mu_n$ contaminated by additive noise. The relationship between $\epsilon_n$ and $\mu_n$ determines the signal-to-noise ratio (SNR), and characterizes when the hypotheses can be distinguished with vanishing probability of error. In the general case, $\f_{0,n}(x)$ can be thought of as the noise and $\f_{1,n}(x)$ as the signal distribution. 

We define the \emph{probability of false alarm} for a hypothesis test $\delta_n$ between $\Hyp_{0,n}$ and $\Hyp_{1,n}$ as 
\begin{equation}
\Pro_{\rm FA}(n) \triangleq \Pro_{0,n} [ \delta_n = 1 ]
\end{equation}
and the \emph{probability of missed detection} as
\begin{equation}
\Pro_{\rm MD}(n) \triangleq \Pro_{1,n} [ \delta_n = 0 ].
\end{equation}

A sequence of hypothesis tests $\{\delta_n\}$ is \emph{consistent} if $\Pro_{\rm FA}(n),\Pro_{\rm MD}(n) \to 0$ as $n \to \infty$. We say we have a \emph{rate characterization} for a sequence of consistent hypothesis tests $\{\delta_n\}$ if we can write
\begin{equation}
\lim_{n \to \infty} \frac{\log \pfa(n)}{g_0(n)} = -c,~~~ \lim_{n \to \infty} \frac{\log \pmd(n)}{g_1(n)} = -d,
\end{equation}
where $g_0(n),g_1(n) \to \infty$ as $n \to \infty$ and $0 < c,d<\infty$. The rate characterization describes decay of the error probabilities for large sample sizes. All logarithms are natural. For the problem of testing between i.i.d. samples from two fixed distributions, $g_0(n)=g_1(n)=n$, and $c,d$ are called the \emph{error exponents} \cite{cover}. In the mixture detection problem, $g_0(n)$ and $g_1(n)$ will be sublinear functions of $n$. 

The log-likelihood ratio between the corresponding probability measures of $\Hyp_{1,n}$ and $\Hyp_{0,n}$ is
\begin{equation}
\text{LLR}(n)=\sum_{i=1}^n \log\big(1-\epsilon_n + \epsilon_n \LR_n(X_i)\big). \label{eq:llr}
\end{equation}
In order to perform an \emph{oracle rate} characterization for the mixture detection problem, we consider the sequence of oracle likelihood ratio tests (LRTs) between $\Hyp_{0,n}$ and $\Hyp_{1,n}$ (i.e. with $\epsilon_n, \f_{0,n}, \f_{1,n}$ known):
\begin{equation}
\delta_n (X_1,\ldots,X_n) \triangleq \begin{cases} 1 & \text{LLR}(n) \geq 0 \\0 & \text{otherwise} \end{cases}. \label{eq:lrt}
\end{equation}
It is well known that \eqref{eq:lrt} is optimal for testing between $\Hyp_{0,n}$ and $\Hyp_{1,n}$ in the sense of minimizing $\frac{\pfa(n)+\pmd(n)}{2}$, which is the average probability of error when the null and alternative are assumed to be equally likely \cite{lehmann,poor}. It is valuable to analyze $\pfa(n)$ and $\pmd(n)$ separately since many applications incur different costs associated with false alarms and missed detections. 

{\bf Location Model:} The detectable region for a location model is the set of sequences $\{(\epsilon_n, \mu_n)\}$ such that a sequence of consistent oracle tests $\{\delta_n\}$ exist. 
For convenience of analysis, we introduce the parameterization 
\begin{equation}
\epsilon_n = n^{-\beta}
\end{equation}
where $\beta \in (0,1)$ as necessary.  Following the terminology of \cite{castro}, when $\beta \in (0,\frac{1}{2})$, the mixture is said to be a ``dense mixture''. If $\beta \in (\frac{1}{2},1)$, the mixture is said to be a ``sparse mixture''.

\subsection{Related Work}\label{sec:rw} Prior work on mixture detection has been focused primarily on the Gaussian location model. The main goals in these works have been to determine the detectable region and construct \emph{optimally adaptive} tests (i.e. those which are consistent independent of knowledge of $\{(\epsilon_n$, $\mu_n)\}$, whenever possible). The study of detection of mixtures where the mixture probability tends to zero was initiated by Ingster for the Gaussian location model \cite{ingster}. Ingster characterized the detectable region, and showed that outside the detectable region the sum of the probabilities of false alarm and missed detection is bounded away from zero for any test. Since the generalized likelihood statistic tends to infinity under the null, Ingster developed an increasing sequence of simple hypothesis tests that are optimally adaptive. 

Donoho and Jin introduced the Higher Criticism test, which is optimally adaptive and is computationally efficient relative to Ingster's sequence of hypothesis tests, and also discussed some extensions to Subbotin distributions and $\chi^2$-distributions \cite{donoho}. Cai et al. extended these results to the case where $\f_{0,n}(x)$ is standard normal and $\f_{1,n}(x)$ is a normal distribution with positive variance, derived limiting expressions for the distribution of $\text{LLR}(n)$ under both hypotheses, and showed that the Higher Criticism test is optimally adaptive in this case \cite{caijengjin}. Jager and Wellner proposed a family of tests based on $\phi$-divergences and showed that they attain the full detectable region in the Gaussian location model \cite{jager}. Arias-Castro and Wang studied a location model where $\f_{0,n}(x)$ is some fixed but unknown symmetric distribution, and constructed an optimally adaptive test that relies only on the symmetry of the distribution when $\mu_n >0$ \cite{castro}. In a separate paper, Arias-Castro and Wang also considered mixtures of Poisson distributions and showed the problem had similar detectability behavior to the Gaussian location model \cite{castropoisson}. 

Cai and Wu gave an information-theoretic characterization of the detectable region via an analysis of the sharp asymptotics of the Hellinger distance for a wide variety of distributions, and established a strong converse result showing that reliable detection is impossible outside the detectable region in many cases if \eqref{eq:lrt} is not consistent \cite{caiwu}. This work also gave general conditions for the Higher Criticism test to be consistent. {Our work complements \cite{caiwu} by providing conditions for consistency (as well as asymptotic estimates of error probabilities) for optimal tests, with simple to verify conditions for a fairly general class of models. } While the Hellinger distance used in \cite{caiwu} provides bounds on $\Pro_{\rm FA}(n)+ \Pro_{\rm MD}(n)$ for the test specified in \eqref{eq:lrt}, our analysis treats $\Pro_{\rm FA}(n),\Pro_{\rm MD}(n)$ separately as they may have different rates at which they tend to zero and different acceptable tolerances in applications. As we will show in Sec. \ref{sec:glma} and Sec. \ref{sec:rtadap}, there are cases where $\Pro_{\rm FA}(n) \gg \Pro_{\rm MD}(n)$ for adaptive tests and $\Pro_{\rm FA}(n) \ll \Pro_{\rm MD}(n)$ for an oracle test.  

Walther numerically showed that while the popular Higher Criticism statistic is consistent, there exist optimally adaptive tests with significantly higher power for a given sample size at different sparsity levels \cite{walther}. Our work complements \cite{walther} by providing a benchmark to meaningfully compare the sample size and sparsity trade-offs of different tests with an oracle test. It should be noted that all of the work except \cite{castro,caijengjin} has focused on the case where $\beta > \frac{1}{2}$, and no prior work has provided an analysis of the \emph{rate} at which $\Pro_{\rm FA}(n),\Pro_{\rm MD}(n)$ can be driven to zero with sample size.

\section{Main Results for Rate Analysis}
\subsection{General Case}
Our main result is a characterization of the oracle rate via the test given in \eqref{eq:lrt}. The sufficient conditions required for the rate characterization are applicable to a broad range of parameters in the Gaussian location model (Sec.~\ref{sec:glma}).

We first look at the behavior of ``weak signals'', where $\LR_n$ has suitably controlled tails under the null hypothesis. In the Gaussian location model in Sec.~\ref{sec:glma}, this theorem is applicable to small detectable $\mu_n$.
\begin{theorem} \label{mthm_main}
Let $\gamma_0 \in (0,1)$ and assume that for all $\gamma\in(0,\gamma_0)$ the following conditions are satisfied:
\begin{gather}
\lim_{n \to \infty} \Exp_0 \left[ \frac{(\LR_n -1)^2}{D_n^2}\ind{\LR_n \geq 1 + \frac{\gamma}{\epsilon_n}}\right] =0 \label{eq:c1}\\
\epsilon_n D_n \to 0 \label{eq:c2}\\
\sqrt{n} \epsilon_n D_n \to \infty \label{eq:c3}
\end{gather}
where 
\begin{equation} 
D_n^2 = \Exp_0 [(\LR_n-1)^2]< \infty. 
\end{equation}
Then for the test specified by \eqref{eq:lrt},
\begin{equation}
\lim_{n \to \infty} \frac{\log \Pro_{\rm FA}(n)}{n \epsilon_n^2 D_n^2} = -\frac{1}{8}. \label{eq:mainthm} \end{equation}
Moreover, \eqref{eq:mainthm} holds if we replace $\Pro_{\rm FA}(n)$ with $\Pro_{\rm MD}(n)$. 
\end{theorem}

The quantity $D_n^2$ is known as the $\chi^2$-divergence between $\f_{0,n}(x)$ and $\f_{1,n}(x)$ \cite{gibbssu}. In contrast to the problem of testing between i.i.d. samples from two fixed distributions \cite{dz}, the rate is not characterized by the Kullback-Leibler divergence for the mixture detection problem. 
\begin{proof} We provide a sketch of the proof for $\Pro_{\rm FA}(n)$, and leave the details to Supplemental Material. We first establish that
\begin{equation}
\limsup_{n \to \infty} \frac{\log \Pro_{\rm FA}(n)}{n \epsilon_n^2 D_n^2} \leq -\frac{1}{8} \label{eq:cub0}
\end{equation}
By the Chernoff bound applied to $\Pro_{\rm FA}(n)$ and noting $X_1,\ldots,X_n$ are i.i.d.,
\begin{align*}
\Pro_{\rm FA}(n) &= \Pro_0\big[ \text{LLR}(n) \geq 0\big] \leq  \left( \min_{0\leq s \leq 1} \Exp_0 \big[ \big(1- \epsilon_n + \epsilon_n \LR_n(X_1) \big)^s\big] \right)^n \\
&\leq \left(\Exp_0 \left[ \sqrt{1- \epsilon_n + \epsilon_n \LR_n(X_1) }\right] \right)^n \numberthis \label{eq:cub}
\end{align*}
By direct computation, we see $\Exp_0[\LR_n(X_1) -1] = 0$, and the following sequence of inequalities hold:
\begin{multline*}
\Exp_0 \left[ \sqrt{1- \epsilon_n + \epsilon_n \LR_n(X_1) }\right] = 1 - \frac{1}{2} \Exp_0\left[ \frac{\epsilon_n^2 (\LR_n(X_1) - 1)^2}{\big(1+\sqrt{1+\epsilon_n (\LR_n(X_1) -1)}\big)^2}\right] \\
\leq 1 - \frac{\epsilon_n^2}{2} \Exp_0\left[ \frac{(\LR_n(X_1) - 1)^2}{(1+\sqrt{1+\epsilon_n (\LR_n(X_1) -1)})^2}\ind{\epsilon_n (\LR_n(X_1) - 1) \leq \gamma} \right] \\
\leq   1 - \frac{\epsilon_n^2 D_n^2}{2(1+\sqrt{1+\gamma})^2} \Exp_0\left[ \frac{(\LR_n(X_1) - 1)^2}{D_n^2}\ind{\LR_n(X_1)\leq1+ \frac{\gamma}{\epsilon_n}} \right]\\ 
=  1 - \frac{\epsilon_n^2 D_n^2}{2(1+\sqrt{1+\gamma})^2} \left(1 - \Exp_0\left[ \frac{(\LR_n(X_1) - 1)^2}{D_n^2}\ind{\LR_n(X_1)\geq1+ \frac{\gamma}{\epsilon_n}} \right]\right) 
\end{multline*}
Since the expectation in the previous line tends to zero by \eqref{eq:c1}, for sufficiently large $n$ it will become smaller than $\gamma$. Therefore we have by \eqref{eq:cub}
$$ 
\frac{\log \Pro_{\rm FA}(n)}{n} \leq \log \left( 1 - \frac{1}{2} \frac{\epsilon_n^2 D_n^2}{(1+ \sqrt{1+\gamma})^2} (1-\gamma) \right).
$$
Dividing both sides by $\epsilon_n^2 D_n^2$ and taking the $\limsup$ using \eqref{eq:c2},\eqref{eq:c3} establishes $\limsup_{n \to \infty} \frac{\log \Pro_{\rm FA}(n)}{n \epsilon_n^2 D_n^2} \leq -\frac{1}{2} \frac{1- \gamma}{(1+\sqrt{1+\gamma})^2}$. Since $\gamma$ can be arbitrarily small, \eqref{eq:cub0} is established.

We now establish that \begin{equation}
\liminf_{n \to \infty} \frac{\log \Pro_{\rm FA}(n)}{n \epsilon_n^2 D_n^2} \geq -\frac{1}{8}. \label{eq:clb0}
\end{equation} The proof of \eqref{eq:clb0} is similar to that of Cramer's theorem (Theorem I.4, \cite{denhollander}). The key difference from Cramer's theorem is that $\text{LLR}(n)$ is the sum of i.i.d. random variables for each $n$, but the distributions of the summands defining $\text{LLR}(n)$ in \eqref{eq:llr} change for each $n$ under either hypothesis. Thus, we modify the proof of Cramer's theorem by introducing a $n$-dependent tilted distribution, and replacing the standard central limit theorem (CLT) with the Lindeberg-Feller CLT for triangular arrays  (Theorem 3.4.5, \cite{durrett}).

We introduce the tilted distribution $\tilde{\f}_n(x)$ corresponding to $\f_{0,n}(x)$ by
\begin{equation}
 \tilde{\f}_n(x) = \frac{\big(1-\epsilon_n + \epsilon_n \LR_n(x)\big)^{s_n}}{\Lambda_n(s_n)} \f_{0,n}(x) \label{eq:tilt}
\end{equation}
where $\Lambda_n(s) = \Exp_0\big[\big(1-\epsilon_n+\epsilon_n \LR_n(X_1)\big)^s\big]$, which is convex with $\Lambda_n(0)=\Lambda_n(1)=1$, and $s_n = \argmin_{0\leq s \leq 1} \Lambda_n(s)$. Let $\tilde{\Pro},\tilde{\Exp}$ denote the tilted measure and expectation, respectively (where we suppress the $n$ for clarity). A standard dominated convergence argument (Lemma 2.2.5, \cite{dz}) shows that
\begin{equation}
\tilde{\Exp}\big[ \log\big(1-\epsilon_n + \epsilon_n \LR_n(X_1)\big) \big] =0.
\end{equation}
Define the variance of the log-likelihood ratio for one sample as 
\begin{equation}
\sigma_n^2 = \tilde{\Exp} \left[ \big(\log\big(1+ \epsilon_n ( \LR_n(X_1) -1 ) \big) \big)^2\right].
\label{eq:sigma2}
\end{equation}
For sufficiently large $n$ such that Lemma \ref{suplem1} (proved in Supplementary Material) holds, namely that $C_1 \epsilon_n^2 D_n^2 \geq \sigma_n^2 \geq C_2 \epsilon_n^2 D_n^2$, we have:
\begin{align}
\pfa(n)&= \Pro_0 \left[\LLR(n) \geq 0\right]= \Exp_0 \left[\ind{ \LLR(n) \geq 0 }\right]\nonumber \\
&=  \left(\Lambda_n(s_n)\right)^n \tilde{\Exp} \left[ e^{- \LLR(n)} \ind{  \LLR(n) \geq 0 }\right]\nonumber \\
&=  \left(\Lambda_n(s_n)\right)^n \tilde{\Exp} \left[ e^{- \LLR(n)} |  \LLR(n) \geq 0\right] \tilde{\Pro}\left[\LLR(n) \geq 0\right] \nonumber\\
&\geq  \left(\Lambda_n(s_n)\right)^n e^{-\tilde{\Exp} \left[ \LLR(n) |  \LLR(n) \geq 0\right]} \tilde{\Pro}\left[\LLR(n) \geq 0\right]\label{eq:b0}\\
&=  \left(\Lambda_n(s_n)\right)^n e^{-\frac{\tilde{\Exp} \left[ \LLR(n) \ind{ \LLR(n) \geq 0 }\right]}{\tilde{\Pro}\left[\LLR(n)\geq 0\right]}} \tilde{\Pro}\left[\LLR(n) \geq 0\right]\nonumber\\
&\geq  \left(\Lambda_n(s_n)\right)^n e^{-\frac{\tilde{\Exp} \left[ |\LLR(n)|\right]}{\tilde{\Pro}\left[\LLR(n)\geq 0\right]}} \tilde{\Pro}\left[\LLR(n) \geq0\right] \label{eq:b1}\\
&\geq  \left(\Lambda_n(s_n)\right)^n e^{-\frac{\sqrt{\tilde{\Exp} \left[ \left(\LLR(n)\right)^2\right]}}{\tilde{\Pro}\left[\LLR(n)\geq 0\right]}} \tilde{\Pro}\left[\LLR(n) \geq 0\right]\label{eq:b2} \\
&= \left(\Lambda_n(s_n)\right)^n e^{-\frac{\sqrt{n \sigma_n^2}}{\tilde{\Pro}\left[\LLR(n)\geq0\right]}}\tilde{\Pro}\left[\LLR(n) \geq 0\right]\nonumber\\
&\geq \left(\Lambda_n(s_n)\right)^n e^{-\frac{\sqrt{n C_1 \epsilon_n^2 D_n^2}}{\tilde{\Pro}\left[\LLR(n)\geq0\right]}} \tilde{\Pro}\left[\LLR(n) \geq 0\right]\label{eq:b3}
\end{align}
where \eqref{eq:b0} follows from Jensen's inequality, \eqref{eq:b1} by $\LLR(n)\ind{\LLR(n)>0}$ $\leq|\LLR(n)|$, \eqref{eq:b2} by Jensen's inequality, and \eqref{eq:b3} by Lemma \ref{suplem1} proved in the Supplementary Material.

Taking logarithms and dividing through by $n \epsilon_n^2 D_n^2$ gives
$$
\frac{\log \pfa(n)}{n \epsilon_n^2 D_n^2} \geq
 \frac{\log \Lambda_n(s_n)}{\epsilon_n^2 D_n^2} - \frac{\sqrt{C_1}}{\tilde{\Pro}\left[\LLR(n) \geq 0\right]} \frac{1}{\sqrt{n} \epsilon_n D_n} + \frac{\log \tilde{\Pro}\left[\LLR(n) \geq 0\right]}{n \epsilon_n^2 D_n^2}. 
$$
Taking $\liminf$ and applying Lemma \ref{suplem2}, in which it is established that $\tilde{\Pro}[\LLR(n) \geq 0] \to \frac{1}{2}$, and Lemma \ref{suplem3} in which it is established that\\ $\liminf_{n \to \infty} \frac{\log \Lambda_n(s_n)}{\epsilon_n^2 D_n^2} \geq - \frac{1}{8} $, (see Supplementary Material),
along with the assumption $n \epsilon_n^2 D_n^2 \to \infty$ establishes that $\liminf_{n \to \infty} \frac{\log \pfa(n)}{n \epsilon_n^2 D_n^2} \geq -\frac{1}{8}$. 

The analysis under $\Hyp_{1,n}$ for $\Pro_{\rm MD}(n)$ relies on the fact that the $X_i$ are i.i.d.~with pdf $(1-\epsilon_n + \epsilon_n \LR_n) \f_{0,n}(x)$, which allows the use of $1-\epsilon_n + \epsilon_n \LR_n$ to change the measure from the alternative to the null. The upper bound is established identically, by noting that the Chernoff bound furnishes 
\begin{align*}
\Pro_{\rm MD}(n) &=  \Pro_{1,n}\big[ -\text{LLR}(n) >0 \big] \leq \left(\Exp_{1} \left[ \frac{1}{\sqrt{1- \epsilon_n + \epsilon_n \LR_n(X_1) }} \right] \right)^n \\
& = \left(\Exp_{0} \left[ \sqrt{1- \epsilon_n + \epsilon_n \LR_n(X_1) }\right] \right)^n 
\end{align*}
Similarly, the previous analysis can be applied to show that \eqref{eq:clb0} holds with $\Pro_{\rm FA}(n)$ replaced with $\Pro_{\rm MD}(n)$.
\end{proof}

In order to study the behavior of tests when Thm \ref{mthm_main} does not hold, we rely on the following bounds for $\Pro_{\rm MD}(n),\Pro_{\rm FA}(n)$: 
\begin{theorem} \label{lbthm_main}(a)~Let $\{\delta_n\}$ be any sequence of tests such that 
$$
\limsup_{n \to \infty} \Pro_{\rm FA}(n) < 1,
$$
then,
\begin{equation}
\liminf_{n \to \infty} \frac{ \log \Pro_{\rm MD} (n)} {n \epsilon_n} \geq -1. \label{eq:unilb}
\end{equation}
(b)~The following upper and lower bounds for $\Pro_{\rm FA}(n)$ hold for the test specified by \eqref{eq:lrt}: 
\begin{gather}
\Pro_{\rm FA}(n) \leq 1 - (\Pro_0 [ \LR_n \leq 1 ])^n \label{eq:univfaub}\\
\Pro_{\rm FA}(n) \geq \Pro_0 \left[ \sum_{i=1}^n \log \max \big(1 - \epsilon_n, \epsilon_n \LR_n(X_i) \big) \geq 0 \right]. \label{eq:univfalb}
\end{gather}
\end{theorem}
These bounds are easily proved by noting if all observations under $\Hyp_{1,n}$ come from $\f_{0,n}$, then a miss detection occurs (a), and at least one sample must have $\LR_n\geq 1$ in order to raise a false alarm (b). 

Note that these are universal bounds in the sense that they impose no conditions on $\f_{1,n}(x), \f_{0,n}(x)$ and $\epsilon_n$. Also note that the bound of Thm \ref{lbthm_main}(a) is independent of any divergences between $\f_{0,n}(x)$ and $\f_{1,n}(x)$, and it holds for any consistent sequence of tests because $\Pro_{\rm FA}(n) \to 0$. This is in contrast to the problem of testing between i.i.d. samples from fixed distributions, where the rate is a function of divergence \cite{dz}. 

When the conditions of Thm~\ref{mthm_main} do not hold, we have the following rate characterization for ``strong signals'', where $\LR_n$ is under the $\f_{1,n}(x)$ distribution in an appropriate sense.  In the Gaussian location model in Sec.~\ref{sec:glma}, this theorem is applicable to large detectable $\mu_n$.
\begin{theorem} \label{ubthm} Let $M_0>1$, and assume that for all $M>M_0$, the following condition is satisified: 
	\begin{equation}
\Exp_0 \left[ \LR_n \ind{ \LR_n > 1 + \frac{M}{\epsilon_n}} \right] \to 1 \label{eq:ubcond}.
\end{equation}
Then for the test specified by \eqref{eq:lrt}, 
\begin{gather}
\limsup_{n \to \infty} \frac{\log \Pro_{\rm FA} (n)}{n \epsilon_n} \leq -1 \label{eq:gub}\\
\lim_{n \to \infty}  \frac{\log \Pro_{\rm MD} (n)}{n \epsilon_n}=-1. \label{eq:ssmd}
\end{gather} 
\end{theorem}
\begin{proof}

We first prove \eqref{eq:gub}. Let 
\begin{equation*}
\phi(x) = 1+ s x -(1+x)^s.
\end{equation*}
By Taylor's theorem, we see for $s \in (0,1)$ and $x \geq -1$ that $\phi(x) \geq 0$.
Since $\Exp_0[\LR_n-1]=0$,
\begin{equation*}
 \Exp_0 [ (1- \epsilon_n + \epsilon_n \LR_n(X_1) )^s] = 1 - \Exp_0[\phi(\epsilon_n (\LR_n (X_1) -1))].
\end{equation*}
Note this implies $\Exp_0[\phi(\epsilon_n (\LR_n (X_1) -1))] \in [0,1]$ since $ \Exp_0 [ (1- \epsilon_n + \epsilon_n \LR_n(X_1) )^s]$ is convex in $s$ and is $1$ for $s=0,1$. 
As in the proof of Thm \ref{mthm_main}, by the Chernoff bound,
\begin{equation*}
\Pro_{\rm FA}(n) \leq  \left(\Exp_0 [ (1- \epsilon_n + \epsilon_n \LR_n(X_1) )^s] \right)^n 
\end{equation*}
for any $s \in (0,1)$. Thus, supressing the dependence on $X_1$, and assuming $M>M_0$, we have
\begingroup
\allowdisplaybreaks
\begin{align*}
\frac{\log \Pro_{\rm FA}(n)}{n} &\leq \log \Exp_0 \left[ (1- \epsilon_n + \epsilon_n \LR_n(X_1) )^s\right] \\
&= \log (1-\Exp_0\left[\phi(\epsilon_n (\LR_n -1))\right])\\
&\leq - \Exp_0\left[\phi(\epsilon_n (\LR_n -1))\right] \numberthis \label{eq:ubp1}\\
&\leq  - \Exp_0\left[\phi(\epsilon_n (\LR_n -1)) \ind{ \epsilon_n (\LR_n-1) \geq M }\right] \numberthis \label{eq:ubp2} \\
&= - \Exp_0 \left[ \left( 1 + s \epsilon_n (\LR_n -1) - (1+ \epsilon_n (\LR_n -1))^s \right)  \ind{ \epsilon_n (\LR_n-1) \geq M } \right] \\
&\leq - \Exp_0 \left[ \left( s \epsilon_n (\LR_n -1) - (1+ \epsilon_n (\LR_n -1))^s \right)  \ind{ \epsilon_n (\LR_n-1) \geq M } \right] \\
&\leq - \Exp_0 \left[ \left( s \epsilon_n (\LR_n -1) - 2^s \epsilon_n^s (\LR_n -1)^s \right)  \ind{ \epsilon_n (\LR_n-1) \geq M } \right] \numberthis \label{eq:ubp3}\\
&= - \Exp_0 \left[ \epsilon_n (\LR_n -1)  \left( s - \frac{2^s} { (\epsilon_n (\LR_n -1))^{1-s}} \right)  \ind{ \epsilon_n (\LR_n-1) \geq M } \right]\\
&\leq - \Exp_0 \left[ \epsilon_n (\LR_n -1)  \left( s - \frac{2} { M^{1-s}} \right)  \ind{ \epsilon_n (\LR_n-1) \geq M } \right] \numberthis \label{eq:ubp4}\\
&= - \epsilon_n \left( s - \frac{2}{M^{1-s}} \right) \Exp_0\left[ \LR_n \left(1 - \frac{1}{\LR_n} \right)  \ind{ \epsilon_n (\LR_n-1) \geq M } \right] \\
&\leq  - \epsilon_n \left( s - \frac{2}{M^{1-s}} \right) \Exp_0\left[ \LR_n \left(1 - \frac{1}{1+ \frac{M}{\epsilon_n}} \right)  \ind{ \epsilon_n (\LR_n-1) \geq M } \right] \\
&=  - \epsilon_n \left( s - \frac{2}{M^{1-s}} \right)  \left(1 - \frac{1}{1+ \frac{M}{\epsilon_n}} \right)  \Exp_0\left[ \LR_n \ind{ \epsilon_n (\LR_n-1) \geq M } \right]
\end{align*}
\endgroup
where \eqref{eq:ubp1} follows from $\log(1-x) \leq -x$ for $x\leq 1$, \eqref{eq:ubp2} follows from $\phi(x) \geq 0$, \eqref{eq:ubp3} follows from $(1+x)^s \leq 2^s x^s$ for $x \geq 1$ and taking $M>M_0$, \eqref{eq:ubp4} follows from $s\in (0,1)$.
Dividing both sides of the inequality by $\epsilon_n$ and taking a $\limsup_{n \to \infty}$ establishes
\begin{equation*}
\limsup_{n \to \infty} \frac{\log \Pro_{\rm FA}}{n \epsilon_n} \leq -s + \frac{2}{M^{1-s}}.
\end{equation*}
	Letting $M \to \infty$ and optimizing over $s \in (0,1)$ establishes the \eqref{eq:gub}. By a change of measure between the alternative and null hypotheses, we see that \eqref{eq:gub} also holds with $\Pro_{\rm FA}(n)$ replaced with $\Pro_{\rm MD}(n)$. Combining this with Thm \ref{lbthm_main} establishes \eqref{eq:ssmd}. 
\end{proof}

Theorem \ref{ubthm} shows that the rate of miss detection is controlled by the average number of observations drawn from $\f_{1,n}(x)$ under $\Hyp_{1,n}$, independent of any divergence between $\f_{1,n}(x)$ and $\f_{0,n}(x)$ when \eqref{eq:ubcond} holds. Interestingly, so long as the condition of Thm \ref{ubthm} holds, by Thm \ref{lbthm_main}(a), no non-trivial sequence of tests (i.e. $\limsup_{n \to \infty} \Pro_{\rm FA}(n), \Pro_{\rm MD}(n) < 1$) can achieve a better rate than \eqref{eq:lrt} under $\Hyp_{1,n}$. This is different from the case of testing i.i.d. observations from two fixed distributions, where allowing for a slower rate of decay for $\Pro_{\rm FA}(n)$ can allow for a faster rate of decay for $\Pro_{\rm MD}(n)$ (Sec. 3.4, \cite{dz}). 

In Sec. \ref{sec:glma}, we will show that Thm \ref{ubthm} is not always tight under $\Hyp_{0,n}$, and the true behavior can depend on divergence between $\f_{0,n}(x)$ and $\f_{1,n}(x)$, using the upper and lower bounds of Thm \ref{lbthm_main}(b).

\subsubsection{Comparison to Related Work} Cai and Wu \cite{caiwu} consider a model which is essentially as general as ours, and characterize the detection boundary for many cases of interest, but do not perform a rate analysis. Note that our rate characterization \eqref{eq:mainthm} depends on $D_n$, the $\chi^2$-divergence between $\f_{0,n}$ and $\f_{1,n}$. While the Hellinger distance used in \cite{caiwu} can be upper bounded in terms of the $\chi^2$-divergence, a corresponding lower bound does not exist in general \cite{gibbssu}, and so our results cannot be derived using the methods of \cite{caiwu}. In fact, our results complement \cite{caiwu} in giving precise bounds on the error decay for this problem once the detectable region boundary has been established. Furthermore, as we will show in Thm~\ref{maxtest}, there are cases where the rates derived by analyzing the likelihood ratio test are essentially achievable. 

\subsection{Gaussian Location Model} \label{sec:glma}
In this section, we specialize Thm~\ref{mthm_main} and \ref{ubthm} to the Gaussian location model. The rate characterization proved is summarized in Fig. \ref{fig:glmsummary}. 
We first recall some results from the literature for the detectable region for this model. 
\begin{figure}[!b]
    \centering
    \begin{subfigure}[b]{0.48\textwidth}
        \includegraphics[width=\textwidth]{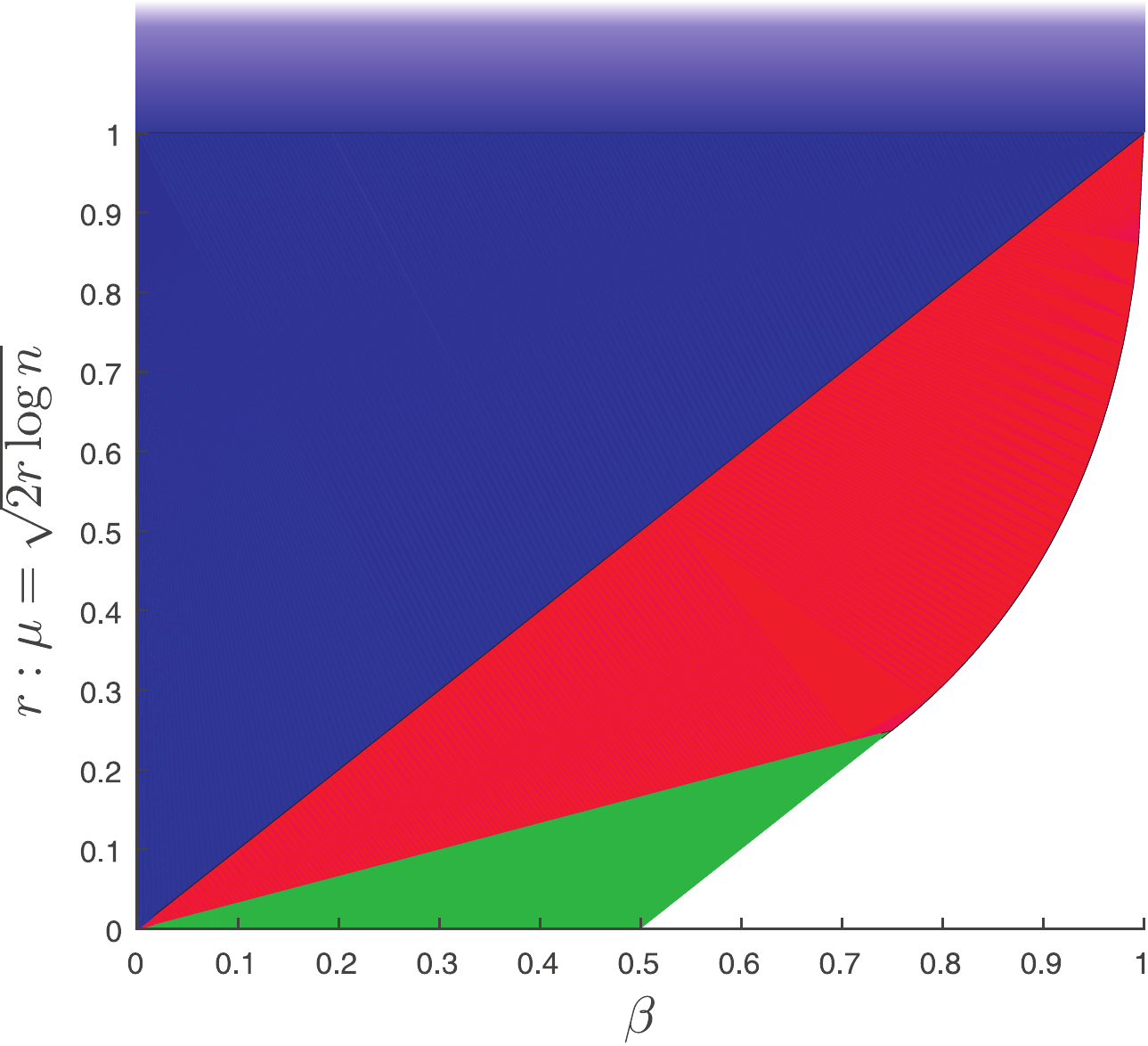}
        \caption{Detectable region ($r$ versus $\beta$) where $\mu_n = \sqrt{2 r \log n}$, $\epsilon_n = n^{-\beta}$}
	\label{fig:glmsummarya}
    \end{subfigure}
    \hfill
    \begin{subfigure}[b]{0.48\textwidth}
        \includegraphics[width=\textwidth]{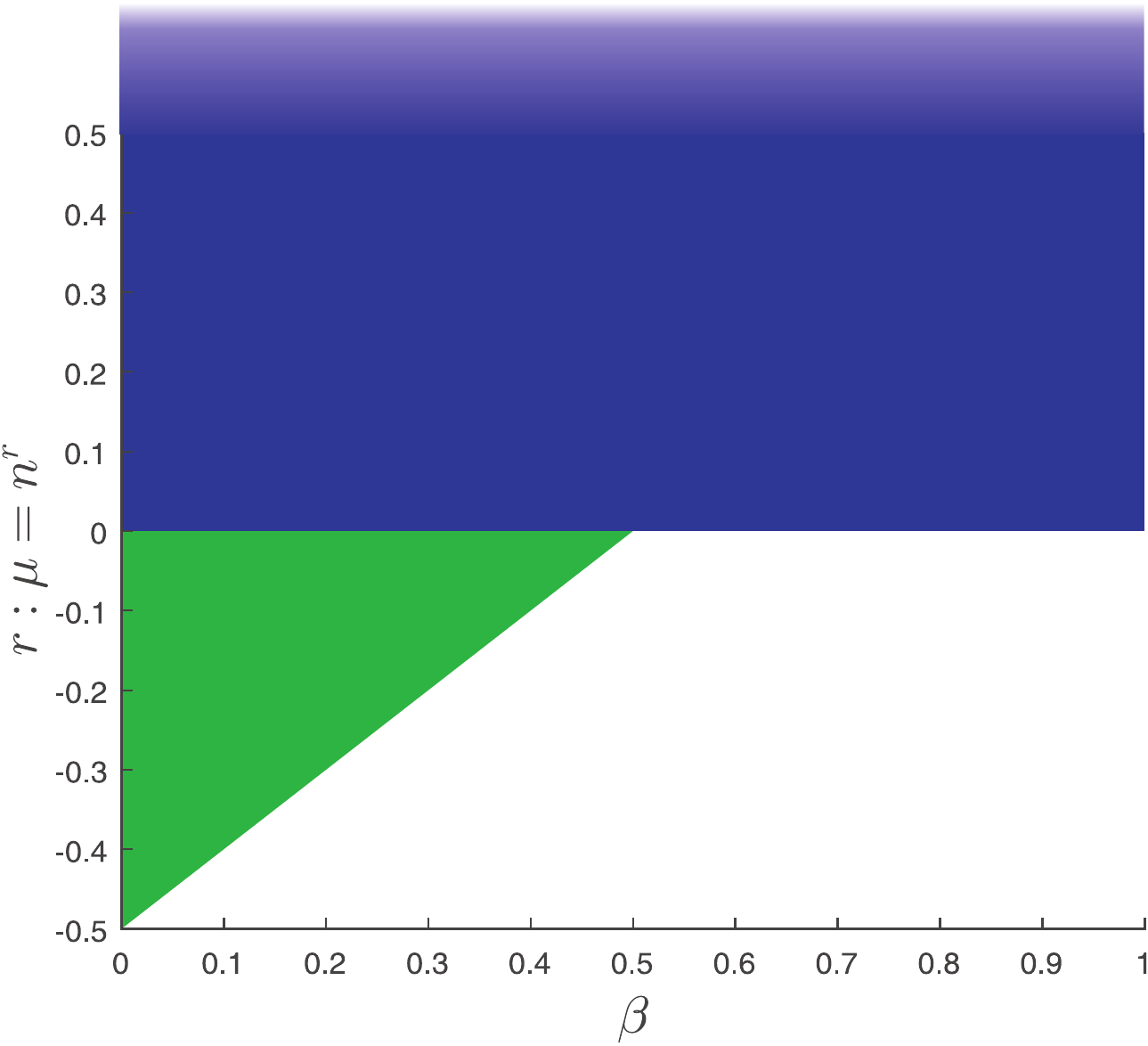}
        \caption{Detectable region ($r$ versus $\beta$) where $\mu_n = n^{r}$, $\epsilon_n = n^{-\beta}$}
    \end{subfigure}
    \caption{Detectable regions for the Gaussian location model. Unshaded regions have $\Pro_{\rm MD}(n)+\Pro_{\rm FA}(n) \to 1$ for any test (i.e. reliable detection is impossible). Green regions are where corollaries \ref{densecase} and \ref{ms} provide an exact rate characterization. The red region is where Thm \ref{gausub_main} provides an upper bound on the rate, but no lower bound. The blue region is where Cor. \ref{ssig} holds, and provides an upper bound on the rate for $\Pro_{\rm FA}(n)$ and an exact rate characterization for $\Pro_{\rm MD}(n)$.} \label{fig:glmsummary} 
\end{figure}

\begin{theorem} \label{dc} The boundary of the detectable region (in $\{(\epsilon_n,\mu_n)\}$ space) is given by (with $\epsilon_n = n^{-\beta})$: 
\begin{enumerate}
\item If $0 < \beta \leq \frac{1}{2}$, then $\mu_{crit,n} = n^{\beta-\frac{1}{2}}$. (Dense)
\item If $\frac{1}{2} < \beta < \frac{3}{4}$, then $\mu_{crit,n} = \sqrt{2 (\beta - \frac{1}{2}) \log n}$. (Moderately Sparse)
\item If $\frac{3}{4} \leq \beta < 1 $, then $\mu_{crit,n} = \sqrt{2 (1- \sqrt{1- \beta})^2 \log n}$. (Very Sparse)
\end{enumerate}
	If in the dense case $\mu_n = n^{r}$, then the LRT \eqref{eq:lrt} is consistent if $r> \beta -\frac{1}{2}$. Moreover, if $r< \beta-\frac{1}{2}$, then $\Pro_{\rm FA}(n)+ \Pro_{\rm MD}(n) \to 1$ for any sequence of tests as $n\to\infty$. If in the sparse cases, $\mu_n = \sqrt{2 r \log n}$, then the LRT is consistent if $\mu_n > \mu_{crit,n}$. Moreover, if $\mu_n < \mu_{crit,n}$, then $\Pro_{\rm FA}(n)+ \Pro_{\rm MD} (n)\to 1$ for any sequence of tests as $n\to\infty$. 
\end{theorem}
\begin{proof}
For the proof see \cite{castro,caijengjin,donoho}.
\end{proof}

We call the set of $\{(\epsilon_n,\mu_n)\}$ sequences where \eqref{eq:lrt} is consistent the \emph{interior of the detectable region}. 
We now begin proving a rate characterization for the Gaussian location model by specializing Thm~\ref{mthm_main}. Note that $\LR_n(x) = e^{\mu_n x-\frac{1}{2}\mu_n^2}$ and $D_n^2 = e^{\mu_n^2}-1$.
A simple computation shows that the conditions in the theorem can be re-written as:

For all $\gamma>0$ sufficiently small:
\begin{gather}
\textstyle
Q\left(-{\frac{3}{2}} \mu_n + \frac{1}{\mu_n}\log \left(1 + \frac{\gamma}{\epsilon_n}\right)\right)+\frac{1}{e^{\mu_n^2}-1}\Big\{Q\left(-\frac{3}{2} \mu_n + \frac{1}{\mu_n}\log \left(1 + \frac{\gamma}{\epsilon_n}\right)\right)\label{eq:c1g}\\
\textstyle
- 2 Q\left(-\frac{1}{2} \mu_n + \frac{1}{\mu_n}\log \left(1 + \frac{\gamma}{\epsilon_n}\right)\right) + Q\left(\frac{1}{2} \mu_n + \frac{1}{\mu_n}\log \left(1 + \frac{\gamma}{\epsilon_n}\right)\right)\Big\} \to 0 \nonumber\\
\epsilon_n^2 (e^{\mu_n^2}-1) \to 0 \label{eq:c2g}
\\
\textstyle
n \epsilon_n^2 (e^{\mu_n^2}-1) \to \infty \label{eq:c3g}
\end{gather}
where $Q(x) = \int_x^\infty \frac{1}{\sqrt{2 \pi}} e^{-\frac{1}{2}x^2} dx$.

\begin{corollary} \label{densecase}(Dense case) If $\epsilon_n = n^{-\beta}$ for $\beta \in (0,\frac{1}{2})$ and $\mu_n = \frac{h(n)}{n^{\frac{1}{2}-\beta}}$ where $h(n)\to \infty$ and $\limsup_{n \to \infty} \frac{\mu_n}{\sqrt{\frac{2 }{3}\beta \log n}} <1$, then \begin{equation}
\lim_{n \to \infty} \frac{\log \Pro_{\rm FA}(n)}{n \epsilon_n^2 (e^{\mu_n^2}-1)} = -\frac{1}{8}. \label{eq:seventeen}
\end{equation}
If $\mu_n \to 0$, \eqref{eq:seventeen} can be rewritten as 
\begin{equation}
\lim_{n \to \infty} \frac{\log \Pro_{\rm FA}(n)}{n \epsilon_n^2 \mu_n^2} = -\frac{1}{8}.
\end{equation}
This result holds when replacing $\Pro_{\rm FA}(n)$ with $\Pro_{\rm MD}(n)$.
\end{corollary}
 
\begin{proof}
It is easy to verify \eqref{eq:c2g} and \eqref{eq:c3g} directly, and \eqref{eq:c1g} if $\mu_n$ does not tend to zero. To verify \eqref{eq:c1g} it suffices to show: If $\mu_n \to 0$, for any $\alpha \in \mathbb{R}$, then $\frac{Q(\alpha \mu_n + \frac{1}{\mu_n}\log(1+\frac{\gamma}{\epsilon_n}))}{e^{\mu_n^2}-1} \to 0$.
Since $e^x -1 \geq x$, it suffices to show that $\frac{Q(\alpha \mu_n + \frac{1}{\mu_n}\log(1+\frac{\gamma}{\epsilon_n}))}{\mu_n^2} \to 0$.
 This can be verified by the standard bound $Q(x) \leq e^{-\frac{1}{2} x^2}$ for $x>0$, and noting that $\alpha \mu_n + \frac{1}{\mu_n}\log(1+\frac{\gamma}{\epsilon_n})>0$ for sufficiently large $n$ and that $\frac{x}{e^{x}-1} \to 1$ as $x \to 0$. 
\end{proof}
The implication of this corollary is that our rate characterization of the probabilities of error holds for a large portion of the 
detectable region up to the detection boundary, as $h(n)$ can be taken such that $\frac{h(n)}{n^{\xi}} \to 0$ for any $\xi>0$, making it negligible with respect to $\mu_{crit,n}$ in Thm~\ref{dc}.

\begin{corollary} \label{ms} (Moderately sparse case) If $\epsilon_n = n^{-\beta}$ for $\beta \in (\frac{1}{2},\frac{3}{4})$ and $\mu_n = \sqrt{ 2(\beta+\frac{1}{2}+\xi) \log n}$ for any $0<\xi < \frac{3-4 \beta}{6}$ then \begin{equation}
\lim_{n \to \infty} \frac{\log \Pro_{\rm FA}(n)}{n \epsilon_n^2 (e^{\mu_n^2}-1)} = -\frac{1}{8}
\end{equation}
and the same result holds replacing $\Pro_{\rm FA}(n)$ with $\Pro_{\rm MD}(n)$.
\end{corollary}
\begin{proof}
It is easy to verify \eqref{eq:c2g} and \eqref{eq:c3g} directly. To verify \eqref{eq:c1g}, note since $Q(\cdot)  \leq 1$ and $\mu_n \to \infty$, we need
\begin{multline*}
\textstyle
\frac{1}{e^{\mu_n^2}-1}\Big\{Q\left(-\frac{3}{2} \mu_n + \frac{1}{\mu_n}\log \left(1 + \frac{\gamma}{\epsilon_n}\right)\right)-\\
\textstyle
2 Q\left(-\frac{1}{2} \mu_n + \frac{1}{\mu_n}\log \left(1 + \frac{\gamma}{\epsilon_n}\right)\right)
 + Q\left(\frac{1}{2} \mu_n + \frac{1}{\mu_n}\log \left(1 + \frac{\gamma}{\epsilon_n}\right)\right) \Big\} \to 0.
\end{multline*}
Thus, it suffices to show  that$Q\big(\!-\frac{3}{2} \mu_n + \frac{1}{\mu_n}\log (1 + \frac{\gamma}{\epsilon_n})\big) \to 0 $, or equivalently, that $-\frac{3}{2} \mu_n + \frac{1}{\mu_n}\log (1 + \frac{\gamma}{\epsilon_n}) \to \infty$ for any fixed $\gamma>0$.  Applying $\log (1 + \frac{\gamma}{\epsilon_n}) \geq \log(\frac{\gamma}{\epsilon_n}) = \beta \log n + \log \gamma$ shows that
\begin{multline*} 
\textstyle
-\frac{3}{2} \mu_n + \frac{1}{\mu_n}\log \big(1 + \frac{\gamma}{\epsilon_n}\big) \geq\\
 -\frac{3}{2} \sqrt{ 2 (\beta -{\textstyle \frac{1}{2}} +\xi) \log n } + \frac{\beta \log n + \log \gamma }{\sqrt{2 (\beta - \frac{1}{2} +\xi) \log n}} \\
= \Big( -{\textstyle\frac{3}{2} }\sqrt{ 2 (\beta - {\textstyle \frac{1}{2}} +\xi) }  + \frac{\beta}{\sqrt{2 (\beta - \frac{1}{2} +\xi) }}\Big) \sqrt{\log n}
+ \frac{ \log \gamma }{\sqrt{2 (\beta - \frac{1}{2} +\xi) \log n}}  \numberthis \label{eq:vsp}
\end{multline*}
where the last term tends to 0 with $n$. Thus, \eqref{eq:vsp} tends to infinity if the coefficient of $\sqrt{\log n}$ is positive, i.e. if $\frac{1}{2} (1- 2 \xi) < \beta <\frac{1}{4} (3-6 \xi)$ , which holds by the definition of $\xi$. Thus, \eqref{eq:vsp} tends to infinity and \eqref{eq:c1g} is proved. 
\end{proof}
Note that $\xi$ can be replaced with an appropriately chosen sequence tending to $0$ such that \eqref{eq:c2g} and \eqref{eq:c3g} hold.
For $\mu_n> \sqrt{\frac{2}{3}\beta\log n}$, \eqref{eq:c1g} does not hold. However, Thm~\ref{ubthm} and Thm~\ref{lbthm_main}  provide a partial rate characterization for the case where $\mu_n$ grows faster than $\sqrt{2 \beta \log n}$ which we present in the following corollary.

\begin{corollary} \label{ssig}If $\epsilon_n = n^{-\beta}$ for $\beta \in (0,1)$ and 
 $\liminf_{n\to \infty} \frac{\mu_n}{ \sqrt{2 \beta \log n}} > 1$, then
\begin{equation}
\lim_{n\to \infty} \frac{ \log \Pro_{\rm MD}(n)}{n \epsilon_n} = -1. 
\end{equation}
If  $\frac{n \epsilon_n}{\mu_n^2} \to \infty$, then 
\begin{equation}
\limsup_{n\to \infty} \frac{ \log \Pro_{\rm FA}(n)}{n \epsilon_n} = -1
\end{equation}
Otherwise,  if $\frac{n \epsilon_n}{\mu_n^2} \to 0$, then
\begin{equation}
\limsup_{n\to \infty} \frac{ \log \Pro_{\rm FA}(n)}{\mu_n^2} \leq -\frac{1}{8}. \label{eq:degeneratebehavior}
\end{equation}
\end{corollary}
\begin{proof}
The condition for Thm~\ref{ubthm} given by \eqref{eq:ubcond} is 
$$
\textstyle
Q\left( \frac{1}{\mu_n}\log\big(1+\frac{M}{\epsilon_n}\big) - \frac{1}{2}\mu_n \right) \to 1.
$$
This holds if $ \frac{1}{\mu_n}\log(1+\frac{M}{\epsilon_n}) - \frac{1}{2}\mu_n \to - \infty$, which is true if $r > \beta$.

To show that $\liminf_{n\to \infty} \frac{ \log \Pro_{\rm FA}(n)}{n \epsilon_n} \geq  -1$ if $\frac{n \epsilon_n}{\mu_n^2} \to \infty$, we can apply a similar argument to the lower bound for Thm \ref{mthm_main} to the lower bound given by \eqref{eq:univfalb} and is thus omitted. Instead, we show a short proof of  $\liminf_{n\to \infty} \frac{ \log \Pro_{\rm FA}(n)}{n \epsilon_n} \geq  -C$ for $C \geq 1$ using \eqref{eq:univfalb}. Note that we can loosen \eqref{eq:univfalb} to 
\begin{equation*}
\Pro_{\rm FA}(n) \geq \Pro_0 \left[ \sum_{i=1}^k \log \left(1 - \epsilon_n \right) + \sum_{i=k+1}^n \log \big( \epsilon_n \LR_n(X_i) \big) \geq 0 \right]
\end{equation*} for any $k$ and explicitly compute a lower bound to $\Pro_{\rm FA}(n)$ in terms of the standard normal cumulative distribution function. Optimizing this bound over the choice of $k$ establishes that $\liminf_{n\to \infty} \frac{ \log \Pro_{\rm FA}(n)}{n \epsilon_n} \geq -C$ for some constant $C \geq 1$ (with $C=1$ if $\frac{\mu_n}{\sqrt{\log n}} \to \infty$). The lower bounding of \eqref{eq:univfalb} in a manner similar to \eqref{mthm_main} recovers the correct constant when $\mu_n$ scales as $\sqrt{2 r \log n}$. 

To see that the log-false alarm probability scales faster than $n \epsilon_n$ when $\frac{n \epsilon_n}{\mu_n^2} \to 0 $, one can apply \eqref{eq:univfaub}. In this case, 
$$
\log \Pro_{\rm FA} (n) \leq \log\left(1 - \big(1- Q({\textstyle\frac{1}{2}}\mu_n)\big)^n\right).
$$
Applying the standard approximation 
\begin{equation}
\frac{x e^{-\frac{1}{2}x^2}}{\sqrt{2 \pi}(1+x^2)} \leq Q(x) \leq \frac{e^{-\frac{1}{2}x^2}}{x \sqrt{2 \pi}} \text{  for } x>0 \label{eq:qapprox},
\end{equation} we see $\limsup_{n\to \infty} \frac{ \log \Pro_{\rm FA}(n)}{\mu_n^2} \leq - \frac{1}{8}$.  
\end{proof}

Note that \eqref{eq:degeneratebehavior} shows an asymmetry between the rates for the miss detection and false alarm probabilities, since there is a fundamental lower bound due to the sparsity under the alternative for the miss probability, but not under the null. 

Theorems \ref{mthm_main} and \ref{ubthm} do not hold when $\epsilon_n = n^{-\beta}$ and $\mu_n = \sqrt{2 r \log n}$ where $r \in ( \frac{\beta}{3}, \beta )$ for $\beta \in (0, \frac{3}{4})$ or $r \in ((1-\sqrt{1-\beta})^2,\beta)$ for $\beta \in (\frac{3}{4},1)$. For the remainder of the detectable region, we have an upper bound on the rate derived specifically for the Gaussian location setting. One can think of this as a case of ``moderate signals''.

\begin{theorem}\label{gausub_main} Let $\epsilon_n = n^{-\beta}$ and $\mu_n = \sqrt{2 r \log n}$ where $r \in \left( \frac{\beta}{3}, \beta \right)$ for $\beta \in (0, \frac{3}{4})$ or $r \in ((1-\sqrt{1-\beta})^2,\beta)$ for $\beta \in (\frac{3}{4},1)$. Then,
\begin{equation}
\limsup_{n \to \infty} \frac{ \log \Pro_{\rm FA} (n)}{n \epsilon_n^2 e^{\mu_n^2} \Phi\left( \big( \frac{\beta}{2r} - \frac{3}{2} \big) \mu_n \right)} \leq -\frac{1}{16}. \label{eq:gausubbd_main}
\end{equation}
where $\Phi(x) = 1-Q(x) = \int_{-\infty}^x \frac{1}{\sqrt{2 \pi}} e^{-x^2/2} dx $ denotes the standard normal cumulative distribution function.

Moreover, \eqref{eq:gausubbd_main} holds replacing $\Pro_{\rm FA}$ with $\Pro_{\rm MD}$.
\end{theorem}
\begin{proof} The proof is based on a Chernoff bound with $s=\frac{1}{2}$. Details are given in the Supplemental Material. 
\end{proof}
It is useful to note that $n \epsilon_n^2 e^{\mu_n^2} \Phi\big( ( \frac{\beta}{2r} - \frac{3}{2} ) \mu_n \big)$ behaves on the order of $\frac{n^{1-2 \beta + 2r - r( 1.5 - \beta/2r)^2}}{\sqrt{2 r \log n}}$ for large $n$ in Thm \ref{gausub_main}. 

\section{Rates and Adaptive Testing in the Gaussian Location Model} \label{sec:rtadap}
No adaptive tests prior to this work have had precise rate characterization. Moreover, optimally adaptive tests for $0<\beta<1$ such as the Higher Criticism (HC) \cite{donoho} test or the sign test of Arias-Castro and Wang (ACW) (\cite{castro}, Sec. 1.4)\footnote{We avoid the use of the acronym CUSUM since it is reserved for the most popular test for the quickest change detection problem in Sequential Analysis.}are not amenable to rate analysis based on current analysis techniques. This is due to the fact that the consistency proofs of these tests follow from constructing functions of order statistics that grow slowly under the null and slightly quicker under the alternative via a result of Darling and Erd\"os \cite{darlingerdos}. We therefore analyze the \emph{max test}:
\begin{equation}
\delta_{\max}(X_1,\ldots,X_n) \triangleq  \begin{cases} 1 & \max_{i=1,\ldots,n} X_i \geq \tau_n \\0 & \text{otherwise} \end{cases} \label{eq:maxtest}
\end{equation}
where $\tau_n$ is a sequence of test thresholds. 

While the max test is not consistent everywhere \eqref{eq:lrt} is \cite{castro,donoho}, it has a few advantages over other tests that are adaptive to all $\{(\epsilon_n,\mu_n)\}$ possible (i.e. optimally adaptive). The first advantage is a practical perspective; the max test requires a linear search and trivial storage complexity to find the largest element in a sample, whereas computing the HC or ACW test requires on the order of $n \log n$ operations to compute the order statistics of a sample of size $n$ (which may lead to non-trivial auxiliary storage requirements), along with computations depending on $Q$-functions or partial sums of the signs of the data. Moreover, the max test has been shown to work in applications such as astrophysics \cite{cayon}. It does not require specifying the null distribution, which allows it to be applied to the Subbotin location models as in \cite{castro}. The second advantage is analytical, as the cumulative distribution function of the maximum of an i.i.d. sample of size $n$ with cumulative distribution function $F(x)$ has the simple form of $F(x)^n$. This also provides a simple way to set the test threshold to meet a pre-specified false alarm probability for a given sample size $n$. As most applications focus on the regime where $\epsilon_n = n^{-\beta}$ for $\beta > \frac{1}{2}$, the following theorem shows the max test provides a simple test with rate guarantees for almost the entire detectable region in this case. 

\begin{theorem}\label{maxtest}
For the max test given by \eqref{eq:maxtest} with threshold $\tau_n = \sqrt{2 \log n}$:

The rate under the null is given by
\begin{equation}
\lim_{n \to \infty} \frac{\log \Pro_{\rm FA}(n)}{\log \log n} = - \frac{1}{2}. \label{eq:ratemaxfa}
\end{equation}

Under the alternative, if $\liminf_{n \to \infty} \frac{\mu_n}{\sqrt{2 (1-\sqrt{1-\beta})^2 \log n}} >1$ with $\epsilon_n = n^{-\beta}$, 
\begin{equation}
\lim_{n \to \infty} \frac{\log \Pro_{\rm MD}(n)}{n \epsilon_n Q(\sqrt{2 \log n} - \mu_n)} = -1. \label{eq:ratemax1md}
\end{equation} 
In particular, if $\liminf_{n \to \infty} \frac{\mu_n}{\sqrt{2 \log n}} >1$, the max test achieves the optimal rate under the alternative
\begin{equation}
\lim_{n \to \infty} \frac{\log \Pro_{\rm MD}(n)}{n \epsilon_n} = -1. \label{eq:ratemax2md}
\end{equation}
Otherwise, the max test is not consistent. 
\end{theorem}

\begin{figure}[!b]
        \includegraphics[width=0.5\textwidth]{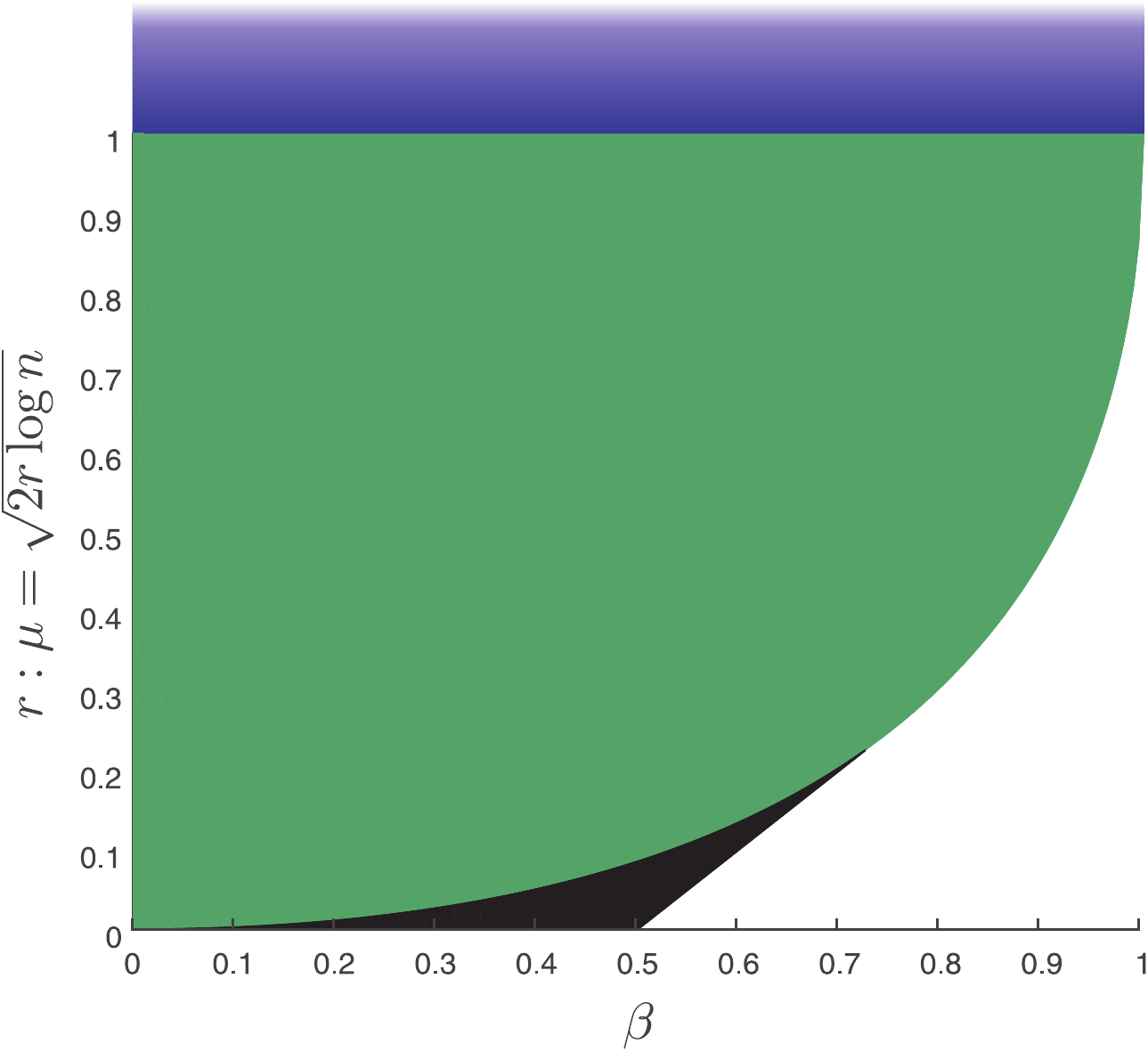}
        \caption{Detectable region of the Max test. White denotes where detection is impossible for any test. Black denotes where the max test is inconsistent. Green denotes where the max test is consistent, but has suboptimal rate under the alternative compared to \eqref{eq:lrt}. Blue denotes where the max test achieves the optimal rate under the alternative. Compare to Fig.~\ref{fig:glmsummarya}.}
\label{fig:maxtest}
\end{figure}
\begin{proof}
The error probabilities for the max test given by \eqref{eq:maxtest} with threshold $\tau_n$
\begin{gather}
\Pro_{\rm FA} (n) = 1-\Phi(\tau_n)^n \label{eq:famax}\\
\Pro_{\rm MD} (n) = \big( (1- \epsilon_n) \Phi(\tau_n) + \epsilon_n \Phi(\tau_n - \mu_n) \big)^n \label{eq:mdmax}
\end{gather}
follow from the cumulative distribution function of the maximum of an i.i.d. sample. The rates \eqref{eq:ratemaxfa},\eqref{eq:ratemax1md},\eqref{eq:ratemax2md} as well as the condition for inconsistency  are derived by applying the approximation \eqref{eq:qapprox} to \eqref{eq:famax} and \eqref{eq:mdmax}.  
\end{proof}

The results of Thm.~\ref{maxtest} are summarized in Fig.~\ref{fig:maxtest}. In particular, if we take $\mu_n = \sqrt{2 r \log n}$ with $r \in \big((1-\sqrt{1-\beta})^2 , 1\big)$, we see $\log \Pro_{\rm MD}(n)$ scales on the order of $\frac{n^{1-\beta-(1-\sqrt{r})^2}}{(1- \sqrt{r}) \sqrt{2 \log n}}$. This is suboptimal compared to the rates achieved by the (non-adaptive) likelihood ratio test \eqref{eq:lrt}, but is of polynomial order (up to a sub-logarithmic factor). Note that the rate of decay of the sum error probability can be slower than that of the miss detection probability, since the false alarm probability is fixed by the choice of threshold, independent of the true $\{(\epsilon_n,\mu_n)\}$ for adaptivity.

\section{Numerical Experiments}

In this section, we provide numerical simulations to verify the rate characterization developed for the Gaussian location model as well as some results comparing the performance of adaptive tests.  

\subsection{Rates for the Likelihood Ratio Test}
\begin{figure}[b!]
\centering
 \begin{subfigure}[t]{0.48\textwidth}
\includegraphics[width=\textwidth]{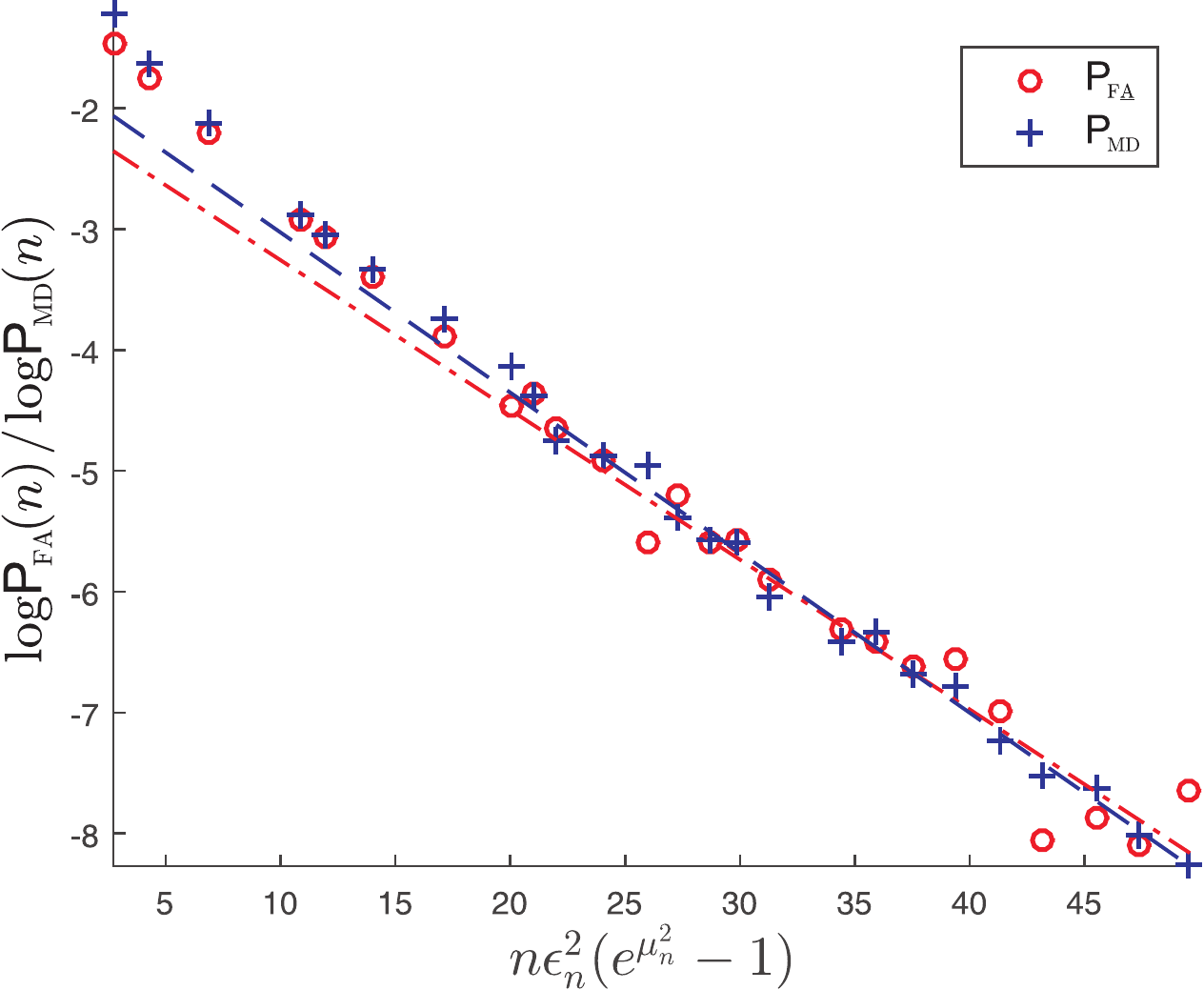}
\caption{Simulations of error probabilities in the Gaussian location model with $\mu_n = 1, \epsilon_n = n^{-0.4}$ for the test \eqref{eq:lrt}. A best fit line for $\log \Pro_{\rm MD}(n)$ is given as a blue dashed line and corresponding line for $\log \Pro_{\rm FA}(n)$ is given as a red dot-dashed line.}
\label{fig:b04m1}
\end{subfigure}
\hfill
\begin{subfigure}[t]{0.48\textwidth}
\includegraphics[width=\textwidth]{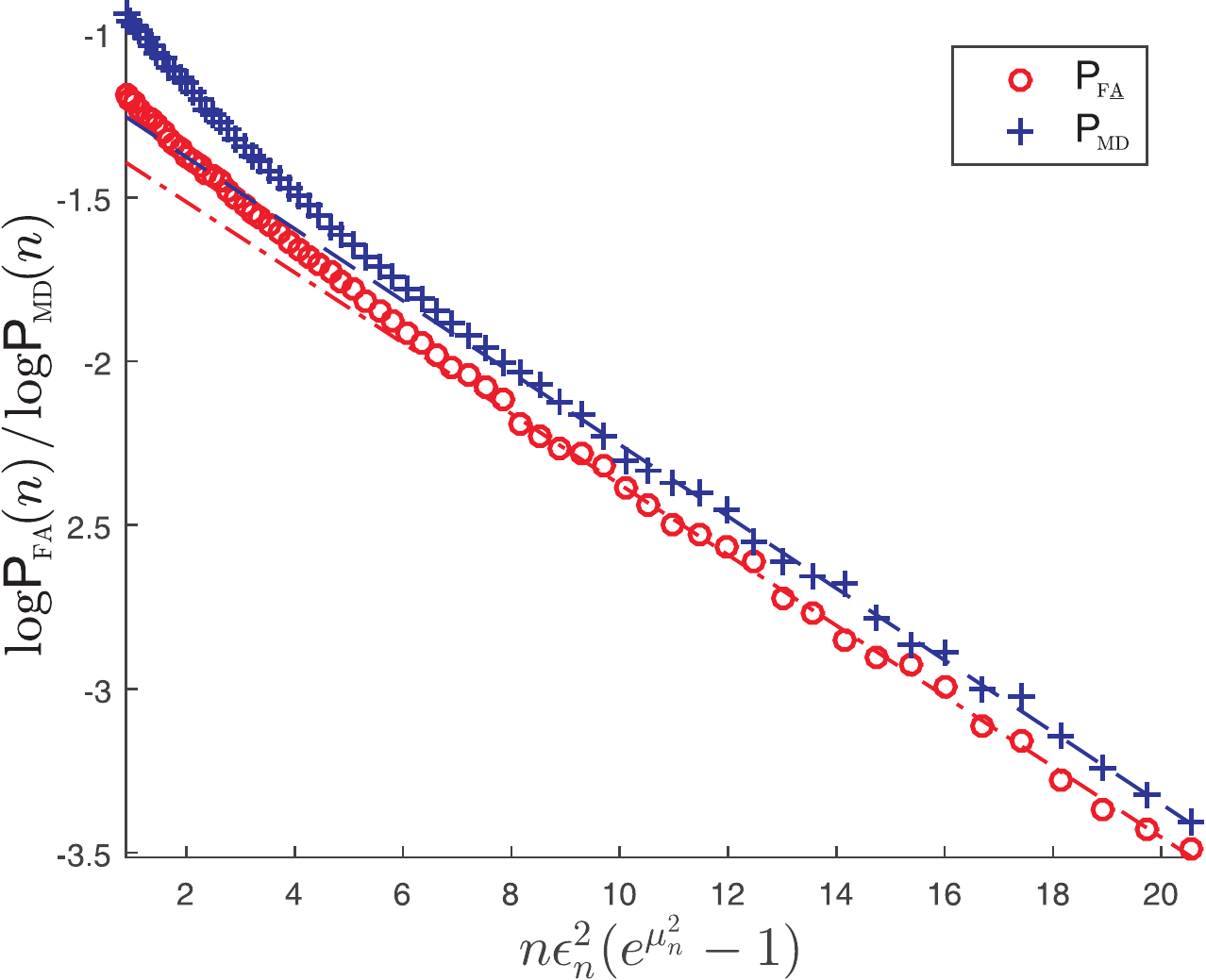}
\caption{Simulations of error probabilities in the Gaussian location model with $\mu_n = \sqrt{2 (0.19) \log n}, \epsilon_n = n^{-0.6}$ for the test \eqref{eq:lrt}. A best fit line for $\log \Pro_{\rm MD}(n)$ is given as a blue dashed line and corresponding line for $\log \Pro_{\rm FA}(n)$ is given as a red dot-dashed line.}
\label{fig:b06}
\end{subfigure}
\caption{Simulation results for Cor. \ref{densecase} and \ref{ms}}
\end{figure}
We first consider the dense case, with $\epsilon_n = n^{-0.4}$ and $\mu_n=1$. The conditions of Cor. \ref{densecase} apply here, and we expect $\frac{\log \Pro_{\rm FA}(n)}{n \epsilon_n^2 (e^{\mu_n^2}-1)} \to -\frac{1}{8}$. Simulations were done using direct Monte Carlo simulation with $10000$ trials for the errors for $n \leq 10^6$. Importance sampling via the hypothesis alternate to the true hypothesis (i.e. $\Hyp_{0,n}$ for simulating $\Pro_{\rm MD}(n)$, $\Hyp_{1,n}$ for simulations $\Pro_{\rm FA}(n)$) was used for $10^6 < n \leq 2 \times 10^7$ with between $10000-15000$ data points. The performance of the test given \eqref{eq:lrt} is shown in Fig. \ref{fig:b04m1}. The dashed lines are the best fit lines between the log-error probabilities and $n \epsilon_n^2 (e^{\mu_n^2}-1)$ using data for $n \geq 350000$. By Cor. \ref{densecase}, we expect the slope of the best fit lines to be approximately $-\frac{1}{8}$. This is the case, as the line corresponding to missed detection has slope $-0.13$ and the line corresponding to false alarm has slope $-0.12$. 

The moderately sparse case with $\epsilon_n = n^{-0.6}$ and $\mu_n = \sqrt{2 (0.19) \log n}$ is shown in Fig. \ref{fig:b06}. The conditions of Cor. \ref{ms} apply here, and we expect $\frac{\log \Pro_{\rm FA}(n)}{n \epsilon_n^2 (e^{\mu_n^2}-1)} \to -\frac{1}{8}$.  Simulations were performed identically to the dense case. The dashed lines are the best fit lines between the log-error probabilities and  $n \epsilon_n^2 (e^{\mu_n^2}-1)$ using data for $n \geq 100000$. By Cor. \ref{ms}, we expect the slope of the best fit lines to be approximately $-\frac{1}{8}$. Both best fit lines have slope of $-0.11$. It is important to note that $\Pro_{\rm FA}(n),\Pro_{\rm MD}(n)$ are both large even at $n=2 \times 10^{7}$ and simulation to larger sample sizes should show better agreement with Cor. \ref{ms}. 

\subsection{Adaptive Testing}
In order to implement an adaptive test, the threshold for the test statistic must be chosen in order to achieve a target false alarm probability. This can be done analytically for the max test by inverting \eqref{eq:famax}. For other tests, which do not have tractable expressions for the false alarm probability, we set the threshold by simulating the test statistic under the null. The threshold is chosen such that the empirical fraction of exceedances of the threshold matches the desired false alarm. As expected, the adaptive tests cannot match the rate under the null with non-trivial behavior under the alternative, and therefore we report the results for adaptive tests at the standard $0.05$ and $0.10$ levels. The miss detection probabilities reported for the max test were computed analytically via \eqref{eq:mdmax}. Note that the likelihood ratio test \eqref{eq:lrt} with threshold set to meet a given false alarm level is the oracle test which minimizes the miss detection probability \cite{dz}. 

As multiple definitions of the Higher Criticism test exist in literature, we use the following version from \cite{caijengjin}: Given a sample $X_1, \ldots, X_n$, let $p_i = Q(X_i)$ for $1\leq i \leq n$. Let $\{p_{(i)}\}$ denote $\{p_i\}$ sorted in ascending order. Then, the higher criticism statistic is given by 
\begin{equation}
\text{HC}^*_n = \max_{1 \leq i \leq n} \text{HC}_{n,i} \text{ where } \text{HC}_{n,i} = \frac{ \frac{i}{n} - p_{(i)}}{\sqrt{p_{(i)} (1- p_{(i)})}} \sqrt{n}
\end{equation}
and the null hypothesis is rejected when $\text{HC}^*_n$ is large. The HC test is optimally adaptive, i.e. is consistent whenever \eqref{eq:lrt} is. 

The ACW test \cite{castro} is implemented as follows: Given the samples $X_1,\ldots,X_n$, let $X_{[i]}$ denote the $i$-th largest sample by absolute value. Then,
\begin{equation}
S^*= \max_{1\leq k \leq n} \frac{\sum_{i=1}^k \text{sgn}(X_{[i]})}{\sqrt{k}}
\end{equation}
and the null hypothesis is rejected when $S^*$ is large. The ACW test is adaptive for $\beta>\frac{1}{2}$. It is unknown how the ACW test behaves for $\beta\leq \frac{1}{2}$. Note that like the Max test (and unlike the HC test), the ACW test does not exploit exact knowledge of the null distribution (but assumes continuity and symmetry about zero). 

\begin{table}[!b]
\footnotesize
	\begin{tabular}{|c|c|c|}
\hline
	\multicolumn{3}{|c|}{LRT}\\
\hline
$n$ & $\pfa(n)$ & $\pmd(n)$  \\
\hline
$10$   & 0.307 & 0.388 \\
$10^2$ & 0.258 & 0.320 \\
$10^3$ & 0.213 & 0.256 \\
$10^4$ & 0.166 & 0.193 \\
$10^5$ & 0.119 & 0.134 \\
$10^6$ & 0.074 & 0.084 \\
\hline
\end{tabular}
	\caption{Error probabilities for $\mu_n=\sqrt{2(0.19)\log n}$, $\epsilon_n=n^{-0.6}$ for the LRT given by \eqref{eq:lrt}. }
\label{table:spweak}
\vskip0.5cm
\footnotesize
\begin{tabular}{|c|c|c|c|c||c|c|c|c|}
\cline{2-9}
\multicolumn{1}{c|}{~}&\multicolumn{4}{c||}{$\Pro_{\text{FA}}=0.05$}&\multicolumn{4}{c|}{$\Pro_{\text{FA}}=0.10$}\\
\hline
$n$ & LRT & Max & HC & ACW & LRT & Max & HC & ACW \\
\hline
	$10$   & 0.776 & 0.845 & 0.790 &  0.807 & 0.665 & 0.744 & 0.666& 0.706 \\
	$10^2$ & 0.667 & 0.814 &  0.775 & 0.816 & 0.542 & 0.704 & 0.630& 0.722  \\
	$10^3$ & 0.548 & 0.789 &  0.728 & 0.792 & 0.417 & 0.672 & 0.561& 0.653  \\
	$10^4$ & 0.403 & 0.762 &  0.688 & 0.751 & 0.281 & 0.639 & 0.491& 0.617  \\
	$10^5$ & 0.252 & 0.733 &  0.623 & 0.685 & 0.158 & 0.603 & 0.396& 0.539  \\
	$10^6$ & 0.119 & 0.699 &  0.546 & 0.602 & 0.064 & 0.562 & 0.295& 0.446  \\
\hline
\end{tabular}
	\caption{Miss Detection probabilities for $\mu_n=\sqrt{2(0.19)\log n}$, $\epsilon_n=n^{-0.6}$, for False Alarm probability 0.05 and 0.10.}
\label{table:spweakfa}
\end{table}
The performance of test \eqref{eq:lrt} is summarized in Table \ref{table:spweak} with a comparison of adaptive tests in the moderately sparse example from the previous section is given in Table \ref{table:spweakfa}. We used $115000$ realizations of the null and alternative. The sample sizes illustrated were chosen to be comparable with applications of sparse mixture detection, such as the WMAP data in \cite{cayon} which has  $n \approx 7 \times 10^4$. Thus, our simulations provide evidence for both larger and smaller sample sizes than used in practice. We see there is a large gap in performance between the likelihood ratio test \eqref{eq:lrt} and the adaptive tests, but the Higher Criticism test performs significantly better than the Max or ACW tests.

\begin{table}[!t]
\footnotesize
	\begin{tabular}{|c|c|c|}
\hline
	\multicolumn{3}{|c|}{LRT}\\
\hline
$n$ & $\pfa(n)$ & $\pmd(n)$  \\
\hline
$10$   & 1.62e-1 & 2.75e-1 \\
$10^2$ & 6.31e-2 & 1.12e-1 \\
$10^3$ & 7.63e-3 & 1.36e-2 \\
$10^4$ & 5.38e-5 & 8.83e-5 \\
\hline
\end{tabular}
	\caption{Error probabilities for $\mu_n=\sqrt{2(0.66)\log n}$, $\epsilon_n=n^{-0.6}$ for the LRT given by \eqref{eq:lrt}. }
\label{table:spstrong}

\vskip0.5cm
\footnotesize
\begin{tabular}{|c|c|c|c|c||c|c|c|c|}
\cline{2-9}
\multicolumn{1}{c|}{~}&\multicolumn{4}{c||}{$\Pro_{\text{FA}}=0.05$}&\multicolumn{4}{c|}{$\Pro_{\text{FA}}=0.10$}\\
\hline
$n$ & LRT & Max & HC & ACW & LRT & Max & HC & ACW \\
\hline
	$10$   & 4.66e-1  & 5.66e-1 & 7.18e-1   & 5.88e-1 & 3.59e-1 & 4.36e-1 & 3.38e-1 & 5.88e-1 \\
	$10^2$ & 1.28e-1  & 2.56e-1 &  6.24e-1  & 4.80e-1 & 8.45e-2 & 1.61e-1 & 1.07e-1 & 4.80e-1   \\
	$10^3$ & 3.69e-3 & 4.40e-2 &  2.48e-2 & 1.33e-1 & 1.89e-3 & 1.80e-2 & 4.20e-3 & 1.33e-1  \\
	$10^4$ & 2.12e-7 & 8.08e-4 &  $\leq$ 1e-5       & 4.43e-3 & 7.10e-8 & 1.32e-4 & $\leq$ 1e-5 & 1.25e-3  \\
\hline
\end{tabular}
\caption{Miss Detection probabilities for $\mu_n=\sqrt{2(0.66)\log n}$, $\epsilon_n=n^{-0.6}$ for False Alarm probability 0.05 and 0.10. }
\label{table:spstrongfa}
\end{table}
For the case of strong signals, we calibrate as $\mu_n = \sqrt{2 (0.66) \log n}$ for $\epsilon_n = n^{-0.6}$. This corresponds to the rates given by Thm. \ref{ubthm}. The performance of test \eqref{eq:lrt} is summarized in Table \ref{table:spstrong} with a comparison of adaptive tests in the moderately sparse example from the previous section is given in Table~\ref{table:spstrongfa}. Here we used $180000$ realizations of the null and alternative. As even the max test has error probabilities sufficiently small for many applications in this regime at moderate sample sizes (which are still on the order used in applications \cite{cayon}), we only consider sample sizes up to $n=10^4$. We see that in the strong signal case, the likelihood ratio test performs better than the adaptive tests, but all tests produce sufficiently small error probabilities for most applications.

\section{Conclusions and Future Work}

In this paper, we have presented an rate characterization for the error probability decay with sample size in a general mixture detection problem for the likelihood ratio test. In the Gaussian location model, we explicitly showed that the rate characterization holds for most of the detectable region. A partial rate characterization (an upper bound on the rate under both hypotheses and universal lower bound on the rate under $\Hyp_{1,n}$) was provided for the remainder of the detectable region. In contrast to usual large deviations results \cite{cover,dz} for the decay of error probabilities, our results show that the log-probability of error decays sublinearly with sample size. 

There are several possible extensions of this work. One is to provide corresponding lower bounds for the rate in cases not covered by Thm \ref{mthm_main}. Another is to provide a general analysis of the behavior that is not covered by Thm \ref{mthm_main} and \ref{ubthm}, present in Thm \ref{gausub_main} in the Gaussian location model. As noted in \cite{caijengjin}, in some applications it is natural to require $\pfa(n) \leq \alpha$ for some fixed $\alpha>0$, rather than requiring $\pfa(n) \to 0$. While Thm~\ref{dc} shows the detectable region is not enlarged under in the Gaussian location model (and similarly for some general models \cite{caiwu}), it is conceivable that the oracle optimal test which fixes $\pfa(n)$ (i.e. one which compares $\text{LLR}(n)$ to a non-zero threshold) can achieve a better rate for $\pmd(n)$. It is expected that the techniques developed in this paper extend to the case where $\pfa(n)$ is constrained to a level $\alpha$. In the Gaussian location model, the analysis of \eqref{eq:lrt} constrained to level $\alpha$ problem has been studied in \cite{ingster} via contiguity arguments. 

Finally, it is important to develop tests that are amenable to a rate analysis and are computationally simple to implement over $0<\beta<1$. In the case of weak signals in the Gaussian location model, we see that the error probabilities for the likelihood ratio test, which establish the fundamental limit on error probabilities, decay quite slowly even with large sample sizes. In this case, closing the gap between the likelihood ratio test and adaptive tests is important for applications where it is desirable to have high power tests. In the case of strong signals, we see the miss detection probability for even the simplest adaptive test, the max test, are very small for moderate sample sizes at standard false alarm levels so the rate of decay is not as important as the weak signal case for applications.

\begin{supplement}[id=suppA]
\stitle{Supplemental Material for ``Detecting Sparse Mixtures: Rate of Decay of Error Probability''}
\slink[doi]{COMPLETED BY THE TYPESETTER}
\sdatatype{.pdf}
\sdescription{We provide details of proofs of the main theorems.}
\end{supplement}

\printaddresses

\begin{frontmatter}
\thispagestyle{plain}
\title{Supplemental Material for:\\ ``Detecting Sparse Mixtures: Rate of Decay of Error Probability''}
\runtitle{Detecting Sparse Mixtures''}

\begin{aug}
\author{\fnms{Jonathan G.} \snm{Ligo}\thanksref{m1}\ead[label=e1]{ligo2@illinois.edu}},
\author{\fnms{George V.} \snm{Moustakides}\thanksref{m2}\ead[label=e2]{moustaki@upatras.gr}}
\and
\author{\fnms{Venugopal V.} \snm{Veeravalli}\thanksref{m1}
\ead[label=e3]{vvv@illinois.edu}
}

\thankstext{T1}{Supported by the US National Science Foundation under grants CIF\,1514245 and CIF\,1513373.}
\runauthor{J. G. Ligo et al.}

\affiliation{University of Illinois at Urbana-Champaign\thanksmark{m1}, University of Patras\thanksmark{m2} and Rutgers University\thanksmark{m2}}

\address{Coordinated Science Laboratory\\ and\\
Department of Electrical and\\ 
Computer Engineering\\
University of Illinois at\\
Urbana-Champaign\\
Urbana, IL 61801, USA\\
\printead{e1}\\
\phantom{E-mail:\ }\printead*{e3}}

\address{Department of Electrical and\\
Computer Engineering\\
University of Patras \\
26500 Rio, Greece\\
and\\
Department of Computer Science\\
Rutgers University \\
New Brunswick, NJ 08854, USA\\
\printead{e2}
}

\end{aug}

\end{frontmatter}

\setcounter{page}{1}

\section{Weak Signals: Supporting Lemmas}
In this section, we provide the proofs of the lemmas that are necessary for establishing the validity of Theorem \ref{mthm_main}.

\begin{lemma} \label{suplem1} 
Under the assumptions of Theorem \ref{mthm_main}, there exist positive constants $C_1,C_2$ such that for sufficiently large $n$ we have
$$
C_1 \epsilon_n^2 D_n^2 \geq \sigma_n^2 \geq C_2 \epsilon_n^2 D_n^2,
$$ 
where $\sigma_n^2$ is defined in \eqref{eq:sigma2}.
\end{lemma}

\begin{proof}
We first show that for sufficiently large $n$,
\begin{equation}
C_1 \geq \frac{\sigma_n^2 \Lambda_n(s_n) }{\epsilon_n^2 D_n^2}. \label{lemub}
\end{equation}

Note that
\begin{equation}
\left(\log\left(1+x\right)\right)^2 \left(1+x\right)^s \leq 2 x^2\text{ for }s \in \left(0,1\right),x\geq 1. \label{eq:ubfact}
\end{equation}
This follows from $0 \leq \log\left(1+x\right) \leq \sqrt{x}$ for $x\geq 0$ and $1 \leq \left(1+x\right)^s \leq 2 x$ for $x \geq 1$ and $s \in (0,1)$. Also, note $\Lambda_n(0)=\Lambda_n(1)=1$ implying $s_n \in (0,1)$ by convexity of $\Lambda_n$ (Lemma 2.2.5, \cite{dz}). 

For shorthand, we will write $\LR_n = \LR_n\left(X_1\right)$. Then, 
\begin{align*}
\Lambda_n(s_n) \sigma_n^2 &= \Exp_0 \left[ \left(\log \left(1+ \epsilon_n \left( \LR_n -1 \right)\right)\right)^2 {\left(1+\epsilon_n\left(\LR_n-1\right)\right)^{s_n}} \right] \\
&= \Exp_0\left[ \left(\log \left(1+ \epsilon_n \left( \LR_n -1 \right)\right)\right)^2 {\left(1+\epsilon_n\left(\LR_n-1\right)\right)^{s_n}} \ind{  \epsilon_n \left(\LR_n-1\right) >1 }\right] \\
&+ \Exp_0\left[ \left(\log \left(1+ \epsilon_n \left( \LR_n -1 \right)\right)\right)^2 {\left(1+\epsilon_n\left(\LR_n-1\right)\right)^{s_n}} \ind{  \epsilon_n \left(\LR_n-1\right) \leq  1} \right] \numberthis \label{eq:uba}
\end{align*}
We first consider $ \Exp_0\big[ \left(\log \left(1+ \epsilon_n \left( \LR_n -1 \right)\right)\right)^2 {\left(1+\epsilon_n\left(\LR_n-1\right)\right)^{s_n}} \ind{ \epsilon_n \left(\LR_n-1\right) > 1 } \big] $. By \eqref{eq:ubfact}, we have on the event $\left\{ \epsilon_n \left(\LR_n-1\right) > 1 \right\}$ that 
$$ 
\big(\log \left(1+ \epsilon_n \left(\LR_n-1\right) \right)\big)^2{ \big(1+\epsilon_n\left(\LR_n-1\right)\big)^{s_n}}\leq 2 \big(\epsilon_n \left(\LR_n-1\right)\big)^2.
$$
Thus,
\begin{multline}
\Exp_0\big[ \left(\log \left(1+ \epsilon_n \left( \LR_n -1 \right)\right)\right)^2 {\left(1+\epsilon_n\left(\LR_n-1\right)\right)^{s_n}}\ind{ \epsilon_n \left(\LR_n-1\right) > 1 }\big]\\  
\leq \Exp_0\big[ 2 \left(\epsilon_n\left(\LR_n-1\right)\right)^2 \ind{ \epsilon_n \left(\LR_n-1\right) > 1 }\big] 
\leq 2 \epsilon_n^2 \Exp_0\big[\left(\LR_n-1\right)^2\big] 
= 2 \epsilon_n^2 D_n^2.\label{eq:ub1}
\end{multline}

We now consider 
$$ 
\Exp_0\left[ \left(\log \left(1+ \epsilon_n \left( \LR_n -1 \right)\right)\right)^2 {\left(1+\epsilon_n\left(\LR_n-1\right)\right)^{s_n}} \ind{  \epsilon_n \left(\LR_n-1\right) \leq 1 }\right].
$$ 
A simple calculus argument shows that $\left(\log\left(1+x\right)\right)^2 \leq 5 x^2$ for $x \geq - \frac{1}{2}$. Note that since $\LR_n \geq 0$, $-\epsilon_n \leq \epsilon_n \left(\LR_n-1\right)$. Because $\epsilon_n \to 0$, for sufficiently large $n$ we have that $\epsilon_n < \frac{1}{2}$ and $\left(\log \left(1+ \epsilon_n \left( \LR_n -1 \right)\right)\right)^2 \leq 5 \left(\epsilon_n \left(\LR_n-1\right)\right)^2$ holds. Also, $\left(1+\epsilon_n\left(\LR_n-1\right)\right)^{s_n} \leq 2^{s_n} \leq 2$ on the event  $\left\{ \epsilon_n \left(\LR_n-1\right) \leq 1 \right\}$. Thus, 
\begin{multline}
\Exp_0\left[ \left(\log \left(1+ \epsilon_n \left( \LR_n -1 \right)\right)\right)^2 {\left(1+\epsilon_n\left(\LR_n-1\right)\right)^{s_n}} \ind{  \epsilon_n \left(\LR_n-1\right) \leq 1} \right] \\
\leq \Exp_0\left[10\left(\epsilon_n \left(\LR_n-1\right)\right)^2 \ind{ \epsilon_n \left(\LR_n-1\right) \leq 1} \right]
\leq  10\,\Exp_0\left[\left(\epsilon_n \left(\LR_n-1\right)\right)^2\right]
= 10\,\epsilon_n^2 D_n^2.\label{eq:ub3}
\end{multline}
Using \eqref{eq:ub1},\eqref{eq:ub3}  in \eqref{eq:uba}, we see for sufficiently large $n$ that 
\begin{equation*}
\Lambda_n(s_n) \sigma_n^2 \leq 12 \epsilon_n^2 D_n^2
\end{equation*}
establishing \eqref{lemub}.

We now show that
\begin{equation}
C_2 \leq \frac{\sigma_n^2 \Lambda_n(s_n) }{\epsilon_n^2 D_n^2}. \label{lemlb}
\end{equation}
Taking any $\gamma < \frac{1}{2}$, from Equation \eqref{eq:c1} of Theorem \ref{mthm_main}, 
\begin{align*}
\Lambda_n(s_n) \sigma_n^2 &= \Exp_0 \left[ \left(\log \left(1+ \epsilon_n \left( \LR_n -1 \right)\right)\right)^2 {\left(1+\epsilon_n\left(\LR_n-1\right)\right)^{s_n}} \right] \\
&\geq  \Exp_0 \left[ \left(\log \left(1+ \epsilon_n \left( \LR_n -1 \right)\right)\right)^2 {\left(1+\epsilon_n\left(\LR_n-1\right)\right)^{s_n}} \ind{ \epsilon_n \left(\LR_n -1\right) \leq \gamma }\right]\\
&\geq \Exp_0 \left[ \left(\log \left(1+ \epsilon_n \left( \LR_n -1 \right)\right)\right)^2 \left(\frac{1}{2}\right) \ind{ \epsilon_n \left(\LR_n -1\right) \leq \gamma }\right] \numberthis \label{eq:lbst}\\
&\geq \frac{1}{4} \Exp_0 \left[ \left(\epsilon_n \left(\LR_n -1\right)\right)^2 \ind{ \epsilon_n \left(\LR_n -1\right) \leq \gamma }\right] \numberthis \label{eq:lbst1}\\ 
&= \frac{D_n^2}{4} \Exp_0 \left[ \frac{\left(\epsilon_n \left(\LR_n -1\right)\right)^2}{D_n^2} \ind{ \epsilon_n \left(\LR_n -1\right) \leq \gamma }\right]\\
&= \frac{D_n^2}{4} \Exp_0 \left[ \frac{\left(\epsilon_n \left(\LR_n -1\right)\right)^2}{D_n^2} \ind{ \LR_n \leq 1 + \frac{\gamma}{\epsilon_n} }\right]\\
&=  \frac{D_n^2 \epsilon_n^2}{4} \left(1 - \Exp_0 \left[ \frac{\left(\LR_n -1\right)^2}{D_n^2} \ind{  \LR_n \geq 1 + \frac{\gamma}{\epsilon_n} }\right]\right) \numberthis \label{eq:lbst2}
\end{align*} 

Where \eqref{eq:lbst} follows from $(1+\epsilon_n(\LR_n-1))^{s_n} \geq (1-\epsilon_n)^{s_n} \geq1-\epsilon_n\geq \frac{1}{2}$  for sufficiently large $n$, as $s_n \in (0,1)$ and $ \epsilon_n \to 0$. A simple calculus argument shows that $\frac{1}{2} x^2 \leq \left( \log\left(1+x\right)\right)^2$ for $x \in [-\frac{1}{2},\frac{1}{2}]$. This, along the fact that with $ -\frac{1}{2} <  - \epsilon_n \leq \epsilon_n (\LR_n -1) \leq \gamma < \frac{1}{2} $ on the event $\{\epsilon_n (\LR_n -1) \leq \gamma\}$ for sufficiently large $n$ establishes \eqref{eq:lbst1}. The definition of $D_n^2$ furnishes \eqref{eq:lbst2}. Noting that $\Exp_0 [ \frac{(\LR_n -1)^2}{D_n^2} \ind{  \LR_n \geq 1 + \frac{\gamma}{\epsilon_n} }] \to 0$ by the assumptions of Thm 2.1 in the main text, \eqref{lemlb} is established. 

In order to remove the $\Lambda_n(s_n)$ factor from the bounds, note that $\Lambda_n(s_n) \leq \Lambda_n\left(1\right) \leq 1$ and that $\Lambda_n\left(s\right) \geq \left(1-\epsilon_n\right)^s \geq \frac{1}{2}$ for sufficiently large $n$. This along with \eqref{lemub} and \eqref{lemlb} establishes the lemma.

This lemma is established identically under $\Hyp_{1,n}$ by applying a change of measure to $\Pro_{0,n}$ (which replaces $s_n$ with $1-s_n$ in the argument above). 
\end{proof}

\begin{lemma} \label{suplem2}
Under the assumptions of Theorem \ref{mthm_main}, if we use the tilted measure, we have
\begin{equation}
\tilde{\Pro} \big[\LLR\left(n\right) \geq 0\big] \to \frac{1}{2}
\end{equation} 
as $n \to \infty$. 
\end{lemma}

\begin{proof}
For the proof, we will need the Lindeberg-Feller Central Limit Theorem whose validity is demonstrated in Theorem 3.4.5, \cite{durrett}:

\begin{theorem} For each $n$, let $Z_{n,i}$, $1 \leq i \leq n$, be independent zero-mean random variables. 
Suppose
\begin{equation}
\lim_{n \to \infty} \sum_{i=1}^n \Exp\left[Z_{n,i}^2\right] = \sigma^2 >0 \label{eq:condlin}
\end{equation}
and for all $\gamma>0$, 
\begin{equation}
\lim_{n \to \infty} \sum_{i=1}^n \Exp\left[|Z_{n,i}|^2 \ind{ |Z_{n,i}| > \gamma} \right] =0 \label{eq:condl2}
\end{equation}

Then, $S_n = Z_{n,1} + \ldots + Z_{n,n}$ converges in distribution to the normal distribution with mean zero and variance $\sigma^2$ as $n \to \infty$.
\end{theorem}

Let us now continue with the proof of Lemma \ref{suplem2}. We draw i.i.d. $\{X_i\}_{i=1}^n$ from $\Hyp_{0,n}$. Define for $1\leq m\leq n$
\begin{equation}
\xi_{n,i} = \log \left(1 + \epsilon_n \big( \LR_n(X_i) -1\big)\right),~~Z_{n,i}=\frac{\xi_{n,i}}{\sqrt{n} \sigma_n}.
\end{equation}
Note that
\begin{equation}
\sum_{i=1}^n Z_{n,i} = \frac{\LLR(n)}{\sqrt{n} \sigma_n}.
\end{equation}
We show $\sum_{i=1}^n Z_{n,i}$ converges to a standard normal distribution under the tilted measure. As stated in the main text, $\tilde{\Exp}\left[Z_{n,i}\right] = 0$ and $\tilde{\Exp}[Z_{n,i}^2] = \frac{1}{n}$. Thus, \eqref{eq:condlin} is satisfied with $\sigma^2=1$. 

It remains to check \eqref{eq:condl2}. Since for fixed $n$, the $Z_{n,i}$ are i.i.d, it suffices to verify that
$$
\tilde{\Exp}\left[n Z_{n,1}^2\ind{ |Z_{n,1}| > \gamma} \right]
= \tilde{\Exp}\left[\frac{\xi_{n,1}^2}{\sigma_n^2} \ind{ \frac{\xi_{n,1}^2}{ n \sigma_n^2} > \gamma^2 }\right]\to0,
$$
$n \to \infty$. To simplify notation, let $\LR_n= \LR_n\left(X_{1}\right)$. 
By Lemma \ref{suplem1}, it suffices to show that
$$
\tilde{\Exp}\left[\frac{\xi_{n,1}^2}{\epsilon_n^2 D_n^2} \ind{ \frac{\xi_{n,1}^2}{ C_2\epsilon_n^2 D_n^2 } > n\gamma^2 }\right] \to 0
$$
which changing to the $\Pro_0$ measure is equivalent to showing that for $0<\gamma < \gamma_0$
\begin{equation}
\Exp_0\left[\frac{\xi_{n,1}^2}{\epsilon_n^2 D_n^2} \left(1+\epsilon_n \left(\LR_n-1\right)\right)^{s_n} \ind{ \frac{\xi_{n,1}^2}{ C_2\epsilon_n^2 D_n^2 } > n\gamma^2 }\right] \to 0 \label{eq:suffclt}
\end{equation}
since $\Lambda_n(s_n) \in [\frac{1}{2},1]$ for sufficiently large $n$. 

We decompose \eqref{eq:suffclt} into 
\begin{multline}
\Exp_0\left[\frac{\xi_{n,1}^2}{\epsilon_n^2 D_n^2} \left(1+\epsilon_n \left(\LR_n-1\right)\right)^{s_n} \ind{ \frac{\xi_{n,1}^2}{ C_2\epsilon_n^2 D_n^2 } > n\gamma^2 }\right]\\
=\Exp_0\left[\frac{\xi_{n,1}^2}{\epsilon_n^2 D_n^2} \left(1+\epsilon_n \left(\LR_n-1\right)\right)^{s_n} \ind{ \frac{\xi_{n,1}^2}{ C_2\epsilon_n^2 D_n^2 } > n\gamma^2,\epsilon_n (\LR_n -1) > 1  }\right]\\
+\Exp_0\left[\frac{\xi_{n,1}^2}{\epsilon_n^2 D_n^2} \left(1+\epsilon_n \left(\LR_n-1\right)\right)^{s_n} \ind{ \frac{\xi_{n,1}^2}{ C_2\epsilon_n^2 D_n^2 } > n\gamma^2,\epsilon_n (\LR_n -1) \leq 1  }\right]
\label{eq:decomp}
\end{multline}
and show that both parts in \eqref{eq:decomp} tend to zero. For the first part applying \eqref{eq:ubfact} and $\big(\log(1+x)\big)^2 \leq x$ for $x>0$, 
\begin{align*}
&\Exp_0\left[\frac{\xi_{n,1}^2}{\epsilon_n^2 D_n^2} \left(1+\epsilon_n \left(\LR_n-1\right)\right)^{s_n} \ind{ \frac{\xi_{n,1}^2}{ C_2\epsilon_n^2 D_n^2 } > n\gamma^2,\epsilon_n (\LR_n -1) > 1  }\right]\\
&\leq \Exp_0\left[\frac{2 \epsilon_n^2 \left(\LR_n-1\right)^2}{\epsilon_n^2 D_n^2}\ind{ \frac{\epsilon^2_n ( \LR_n -1)^2}{ C_2\epsilon_n^2 D_n^2 } > n\gamma^2, \epsilon_n (\LR_n -1) > 1 }\right] \\
&= 2 \Exp_0\left[\frac{(\LR_n-1)^2}{D_n^2} \ind{ \frac{\LR_n -1}{ D_n\sqrt{C_2} } > \gamma\sqrt{n}, \epsilon_n \left(\LR_n -1\right) > 1 }\right] \\
&\leq 2 \Exp_0\left[\frac{(\LR_n-1)^2}{D_n^2} \ind{  \LR_n> 1+\sqrt{C_2}\sqrt{n} D_n \epsilon_n \frac{\gamma}{\epsilon_n} }\right] \\
&\leq 2 \Exp_0\left[\frac{(\LR_n-1)^2}{D_n^2} \ind{  \LR_n> 1+\frac{\gamma}{\epsilon_n} } \right] \to 0
\end{align*}
where the last inequality holds because from $\sqrt{n} D_n \epsilon_n \to \infty$ we can conclude that for sufficiently large $n$ we have $\sqrt{C_2}\sqrt{n} D_n \epsilon_n\geq1$.

We now show that the second part in \eqref{eq:decomp} tends to zero as well. We observe that 
since $\LR_n \geq 0$ we have $-\epsilon_n \leq \epsilon_n \left(\LR_n -1\right) $. Using $( \log(1+x))^2 \leq 5 x^2$ for $x \geq -\frac{1}{2}$, and that $ (1+\epsilon_n (\LR_n-1))^{s_n} \leq 2^{s_n} \leq 2$ on the event $\{\epsilon_n(\LR_n -1) \leq 1 \}$, we see that for $n$ sufficiently large such that $\epsilon_n < \frac{1}{2}$, 
\begin{align*}
\Exp_0\left[\frac{\xi_{n,1}^2}{\epsilon_n^2 D_n^2} \big(1+\epsilon_n (\LR_n-1)\big)^{s_n} \ind{ \frac{\xi_{n,1}^2}{ C_2\epsilon_n^2 D_n^2 } > n\gamma^2,\epsilon_n (\LR_n -1) \leq 1  }\right]
\hskip-6cm\\
&\leq 10\Exp_0\left[\frac{\epsilon_n^2(\LR_n-1)^2 }{\epsilon_n^2D_n^2}  \ind{ \frac{ 5\epsilon_n^2(\LR_n-1)^2}{C_2\epsilon_n^2 D_n^2} > n \gamma^2, \epsilon_n (\LR_n -1) \leq 1 } \right] \\
&\leq 10 \Exp_0 \left[ \frac{(\LR_n-1)^2}{D_n^2} \ind{  |\LR_n-1| > \sqrt{\frac{C_2}{5}}\sqrt{n}D_n \gamma } \right]\\
&=10 \Exp_0 \left[ \frac{(\LR_n-1)^2}{D_n^2} \ind{ \LR_n>1+ \sqrt{\frac{C_2}{5}}\sqrt{n}D_n \gamma } \right]\\
&\leq 10 \Exp_0 \left[ \frac{(\LR_n-1)^2}{D_n^2} \ind{ \LR_n>1+ \frac{\gamma}{\epsilon_n} } \right]
\end{align*}
The last equality follows from the fact that $\sqrt{n}\epsilon_nD_n\to\infty$, since this implies that 
$\sqrt{n}D_n\to\infty$, which suggests that for large enough $n$ we cannot have $1-\LR_n > \sqrt{\frac{C_2}{5}}\sqrt{n}D_n \gamma$ but only $\LR_n-1 > \sqrt{\frac{C_2}{5}}\sqrt{n}D_n \gamma$. Finally the last inequality is true for large enough $n$ such that $\sqrt{\frac{C_2}{5}}\sqrt{n}D_n\epsilon_n\geq1$, which is always possible since this quantity tends to infinity because of our assumption in \eqref{eq:c3}.
Thus, \eqref{eq:suffclt} holds and the Lindeberg-Feller CLT shows that $\frac{\LLR(n)}{\sqrt{n} \sigma_n} $ converges to a standard normal distribution under the tilted measure.
Therefore, 
\begin{equation}
\tilde{\Pro}\left[\LLR(n) \geq 0\right] = \tilde{\Pro}\left[\frac{\LLR(n)}{\sqrt{n} \sigma_n} \geq 0\right] \to \frac{1}{2}
\end{equation}
as $n\to\infty$ establishing the lemma.

Verifying the Lindeberg-Feller CLT conditions for analyzing $\pmd$ is done by changing from the $\Pro_{1}$ to the $\Pro_{0}$ measure. 
\end{proof}

\begin{lemma} \label{suplem3}
Under the assumptions of Theorem \ref{mthm_main} we have
\begin{equation}
\liminf_{n \to \infty} \frac{\log \Lambda_n(s_n)}{\epsilon_n^2 D_n^2} \geq - \frac{1}{8}.
\end{equation}
\end{lemma}
\begin{proof}
Consider the function $\left(1+x\right)^s$ for $s\in(0,1)$ and $x\in[-\gamma,\gamma]$ where $0<\gamma < 1$. Then
\begin{align}
\left(1+x\right)^s &= 1 + s x + \frac{1}{2} s \left(s-1\right) x^2 + \frac{1}{6} \frac{\left(1-s\right)\left(2-s\right)}{\left(1+\xi\right)^{3-s}} x^3\nonumber \\
&\geq 1+ s x - \frac{1}{8} x^2 -  \frac{1}{3} \frac{\gamma}{\left(1-\gamma\right)^{3}} x^2
=1+ s x - \omega(\gamma) x^2,
\label{eq:taylor}
\end{align}
where we define $\omega(\gamma)=\frac{1}{8} +  \frac{1}{3} \frac{\gamma}{\left(1-\gamma\right)^{3}}$.
The first equality holds for some $\xi \in [-\gamma,\gamma] $ by the mean value form of Taylor's theorem. The inequality is obtained by minimizing the coefficient of $x^2$ while for the last term we observe that since $x\geq -\gamma$ we have $x^3\geq-\gamma x^2$; furthermore $(1+\xi)^{3-s}\geq(1-\gamma)^3$ and $(1-s)(2-s)\leq2$. When we substitute the previous inequalities we obtain the lower bound in \eqref{eq:taylor}.

Using this, we can lower bound $\Lambda_n(s)$ for all $s \in \left(0,1\right)$. Fix $0 < \gamma < 1$. As before, we will use the shorthand $\LR_n=\LR_n\left(X_1\right)$. Then, for sufficiently large $n$ we have $\epsilon_n < \gamma$ suggesting that $-\gamma\leq\epsilon_n(\LR_n-1)$. Therefore using \eqref{eq:taylor} and assuming $n$ sufficiently large we can write
\begin{align*}
\Lambda_n(s) &= \Exp_0\left[\left(1+\epsilon_n \left(\LR_n-1\right)\right)^s\right]\\
&\geq \Exp_0\left[\left(1+\epsilon_n (\LR_n-1)\right)^s \ind{ \epsilon_n \left(\LR_n -1\right) \leq \gamma }\right]\\
&\geq \Exp_0\left[ \left( 1+ s \epsilon_n (\LR_n-1) \right) \ind{ \epsilon_n (\LR_n-1) \leq \gamma }\right] - \omega(\gamma)\epsilon_n^2 D_n^2\\
&= 1 - \omega(\gamma) \epsilon_n^2 D_n^2  - \Exp_0\left[ \left( 1+ s \epsilon_n \left(\LR_n-1\right) \right) \ind{ \epsilon_n \left(\LR_n-1\right) \geq \gamma }\right]\\
&\geq  1 - \omega(\gamma) \epsilon_n^2 D_n^2  - \Exp_0\left[ \left( 1+ \epsilon_n \left(\LR_n-1\right) \right) \ind{ \epsilon_n \left(\LR_n-1\right) \geq \gamma }\right]\\
&\geq 1 - \omega(\gamma) \epsilon_n^2 D_n^2  - \Exp_0\left[ \left(\frac{1}{\gamma^2} + \frac{1}{\gamma} \right) \epsilon_n^2 (\LR_n-1)^2 \ind{ \epsilon_n \left(\LR_n-1\right) \geq \gamma }\right]\\
&\geq 1 - \left(\omega(\gamma) + \frac{2}{\gamma^2}\Exp_0\left[\frac{(\LR_n-1)^2}{D_n^2} \ind{ \LR_n \geq 1 + \frac{\gamma}{\epsilon_n} } \right] \right) \epsilon_n^2 D_n^2\\
&\geq 1 - \left(\omega(\gamma) + \frac{2}{\gamma^2}\gamma^3\right) \epsilon_n^2 D_n^2,
\end{align*}
where in the second equality we used the fact that $\Exp_0[\LR_n -1] =0$ and in the third inequality we replaced the maximum values of $s=1$. In the fourth inequality we used the property that on the set $\{\LR_n \geq 1 + \frac{\gamma}{\epsilon_n}\}$ we can write $\frac{1}{\gamma^2}\epsilon_n^2 (\LR_n-1)^2\geq1$ and
$\frac{1}{\gamma}\epsilon_n^2 (\LR_n-1)^2\geq\epsilon_n (\LR_n-1)$. Finally in the last inequality using the condition of Theorem \ref{mthm_main} and assuming $n$ sufficiently large the expectation becomes smaller than $\gamma^3$.

Using the previous result we obtain
$$
\liminf_{n \to \infty} \frac{\log \Lambda_n(s_n)}{\epsilon_n^2 D_n^2}\geq
\liminf_{n \to \infty}\frac{\log \left( 1 - (\omega(\gamma) + 2\gamma) \epsilon_n^2 D_n^2 \right) }{\epsilon_n^2 D_n^2}
= -(\omega(\gamma) + 2\gamma),
$$
where for the equality we used the limit $\frac{\log(1-x)}{x} \to -1$ as $x \to 0$ and the assumption that $\epsilon_nD_n\to0$. Letting $\gamma \to 0$ establishes the lemma since $\omega(0)=\frac{1}{8}$. 

The proof is identical under $\Hyp_{1,n}$, where the analogue of the lemma is $\liminf_{n \to \infty} \frac{\Lambda_n(1-s_n)}{\epsilon_n^2 D_n^2} \geq -\frac{1}{8}$.
\end{proof}

\section{Moderate Signals: Gaussian Location Model}\label{sec:gaussianub}

Assume $\frac{3}{2} > \frac{\beta}{2r} > \frac{1}{2}$. 
Recall from the proof of Thm \ref{mthm_main}
\begin{equation}
\Pro_{\rm FA}(n) \leq \left(\Exp_0 \left[ \sqrt{1- \epsilon_n + \epsilon_n \LR_n(X_1) }\right] \right)^n \label{eq:gausublo}
\end{equation}
and 
\begin{equation}
\Exp_0 \left[ \sqrt{1- \epsilon_n + \epsilon_n \LR_n(X_1) }\right] = 1 - \frac{1}{2} \Exp_0\left[ \frac{\epsilon_n^2 (\LR_n(X_1) - 1)^2}{\big(1+\sqrt{1+\epsilon_n (\LR_n(X_1) -1)}\big)^2}\right].
\end{equation}

We write the observations as a multiple of $\mu_n$, $X= \alpha \mu_n$. Then, taking $\mu_n = \sqrt{2 r \log n}$, we have 
\begin{equation}
\LR_n (x) = e^{ - \mu_n^2/2 + \mu_n x}= n^{r (2 \alpha -1)}.
\label{eq:rescale}
\end{equation} 
In view of \eqref{eq:rescale},
\begin{equation}
\epsilon_n (\LR_n -1) = n^{r (2 \alpha -1) - \beta} - n^{-\beta}.
\end{equation}
Thus, if $r (2 \alpha -1) - \beta >0$ we have $\epsilon_n (\LR_n -1) \to \infty$ and if $r (2 \alpha -1) - \beta <0$ we have $\epsilon_n (\LR_n -1) \to 0$ as $n \to \infty$. 

Let $\kappa = \frac{\beta}{2r} + \frac{1}{2}$. 
Note $\frac{x^2}{(1+\sqrt{1+x})^2} \geq \frac{x^2}{4} - \frac{x^3}{8}$ for $x \geq -1$. Thus, on the event $\{X_1 < \kappa \mu_n\}$, 
\begin{align*}
\frac{\epsilon_n^2 (\LR_n(X_1) - 1)^2}{\big(1+\sqrt{1+\epsilon_n (\LR_n(X_1) -1)}\big)^2} &\geq \frac{\epsilon_n^2 (\LR_n(X_1) - 1)^2}{4} - \frac{\epsilon_n^3 (\LR_n(X_1) - 1)^3}{8} \\
&\geq \epsilon_n^2 (\LR_n(X_1) - 1)^2 \left( \frac{1}{4} - \frac{(1-n^{-\beta})}{8} \right) \numberthis \label{eq:rsm1}\\
&\geq \frac{ \epsilon_n^2 (\LR_n(X_1) - 1)^2}{8} \numberthis \label{eq:rs0}
\end{align*}
where \eqref{eq:rsm1} follows from $-1 \leq - \epsilon_n \leq \epsilon_n (\LR_n(X_1) - 1) \leq 1-n^{-\beta}$ on $\{X_1 < \kappa \mu_n\}$, and \eqref{eq:rs0} follows from non-negativity of the terms involved. 
Then, 
\begin{align*}
\Exp_0\left[ \frac{\epsilon_n^2 (\LR_n(X_1) - 1)^2}{\big(1+\sqrt{1+\epsilon_n (\LR_n(X_1) -1)}\big)^2}\right]\hskip-4cm&\\ 
&\geq \Exp_0\left[ \frac{\epsilon_n^2 (\LR_n(X_1) - 1)^2}{\big(1+\sqrt{1+\epsilon_n (\LR_n(X_1) -1)}\big)^2} \ind{ X_1 < \kappa \mu_n} \right] \numberthis \label{eq:rs1} \\
&\geq \frac{\epsilon_n^2}{8} \Exp_0 [  (\LR_n(X_1) - 1)^2 \ind{ X_1 < \kappa \mu_n} ] \numberthis \label{eq:rs2} \\
&= \frac{\epsilon_n^2}{8} \left( e^{\mu_n^2} \Phi( (\kappa-2) \mu_n) - 2 \Phi ( (\kappa - 1) \mu) + \Phi (\kappa \mu) \right)\numberthis \label{eq:identdom}
\end{align*}
where \eqref{eq:rs1} follows from non-negativity, \eqref{eq:rs2} follows from \eqref{eq:rs0} and $\Phi$ denotes the standard normal cumulative distribution function.

By the standard approximation $Q(x)\approx\frac{1}{x \sqrt{2 \pi}} e^{-x^2/2}$, we see that the dominant term in \eqref{eq:identdom} is $\epsilon_n^2 e^{\mu_n^2} \Phi( (\kappa-2) \mu_n)/8$ and is of the order of $\frac{n^{-2 \beta +2r - r (1.5-\beta/2r)^2}}{\sqrt{2 r \log n}}$ which tends to zero by our assumption on $\frac{\beta}{2r}$. Thus, as in the proof of Thm \ref{mthm_main}, by \eqref{eq:gausublo}
$$
\frac{ \log \Pro_{\rm FA}(n)}{n}  \leq \log \left( 1 - \frac{\epsilon_n^2}{16}  \left( e^{\mu_n^2} \Phi( (\kappa-2) \mu_n) - 2 \Phi ( (\kappa - 1) \mu) + \Phi (\kappa \mu) \right) \right).
$$
Dividing both sides by $\epsilon_n^2 e^{\mu_n^2}  \Phi( (\kappa-2) \mu_n)$ and taking a $\limsup$ yields
$$
\limsup_{n \to \infty} \frac{ \log \Pro_{\rm FA}(n)}{n\epsilon_n^2 e^{\mu_n^2}  \Phi( (\kappa-2) \mu_n)} \leq -\frac{1}{16}. 
$$

For consistency, it suffices to require $n\epsilon_n^2 e^{\mu_n^2}  \Phi( (\kappa-2) \mu_n) \to \infty$. Thus, it suffices to require $1-2 \beta +2r - r \left(\frac{3}{2}-\frac{\beta}{2r}\right)^2>0$, since $\sqrt{\log n}$ is negligible with respect to any positive power of $n$. Combining the constraints $-2 \beta +2r - r \left(\frac{3}{2}-\frac{\beta}{2r}\right)^2<0, 1-2 \beta +2r - r \left(\frac{3}{2}-\frac{\beta}{2r}\right)^2>0, \frac{3}{2} > \frac{\beta}{2r} > \frac{1}{2}$ gives the desired rate characterization.

The proof for $\Pro_{\rm MD}$ is identical. Note that this bound is likely not tight (even if it has the right order), since we neglected the event $\{X_1 \geq \kappa \mu_n\}$ to form the bound.

\end{document}